%% file: draft.tex
\documentclass[11pt,runningheads,orivec]{llncs}
\usepackage[left=1in, top=1in, right=1in, bottom=1in]{geometry}
\usepackage[T1]{fontenc}
\usepackage{amsmath,amssymb,mathtools}
\usepackage{booktabs,array}
\usepackage{graphicx}
\usepackage{float}
\usepackage{placeins}
\usepackage[colorlinks=true, citecolor=blue, linkcolor=red, urlcolor=black]{hyperref}
\hypersetup{pdftitle={Does AI Help or Harm? Endogenous Information Acquisition with AI Advice},
  pdfauthor={Kexin Chen, Tao Lin, and Jianwei Huang}}
\usepackage{enumitem}
\usepackage{xparse}
\usepackage{xcolor}
\usepackage{bbding}
\newcommand{\X}{\mathcal X}
\newcommand{\Y}{\mathcal Y}
\newcommand{\A}{\mathcal A}
\newcommand{\E}{\mathbb E}
\newcommand{\R}{\mathbb R}
\newcommand{\DeltaX}{\Delta(\X)}
\newcommand{\DeltaY}{\Delta(\Y)}
\newcommand{\cav}{\operatorname{cav}}
\newcommand{\supp}{\operatorname{supp}}
\newcommand{\argmax}{\operatorname*{arg\,max}}

\newcommand{\one}{\mathbf 1}

\providecommand{\qedsymbol}{\ensuremath{\square}}

\RenewDocumentEnvironment{proof}{O{}}%
{\par\noindent\textit{Proof\if\relax\detokenize{#1}\relax\else\ #1\fi.}\ }%
{\hfill\qedsymbol\par}

\newcounter{wineexample}
\renewcommand{\thewineexample}{\arabic{wineexample}}

\NewDocumentEnvironment{wineexample}{O{}}%
{\refstepcounter{wineexample}\par\medskip\noindent\textbf{Example~\thewineexample\if\relax\detokenize{#1}\relax\else\
(#1)\fi.}\ }%
{\par\medskip}

\begin{document}

\title{Does AI Help or Harm?
\\Endogenous Information Acquisition with AI Advice}
\titlerunning{Does AI Help or Harm? Endogenous Information Acquisition with AI Advice}
\author{
Kexin Chen\inst{1}\and
Tao Lin\inst{2}\and
Jianwei Huang\inst{3,4} \Envelope
}
\authorrunning{K. Chen et al.}
\institute{
School of Science and Engineering,
The Chinese University of Hong Kong,
Shenzhen.
\and
School of Data Science,
The Chinese University of Hong Kong,
Shenzhen.
\and
School of Science and Engineering,
Shenzhen Institute of Artificial Intelligence and Robotics for Society,
Shenzhen Key Laboratory of Crowd Intelligence Empowered Low-Carbon Energy Network,
and CSIJRI Joint Research Center on Smart Energy Storage,
The Chinese University of Hong Kong,
Shenzhen.
\and
Shenzhen Loop Area Institute.
\\[0.3em]
\email{kexinchen2@link.cuhk.edu.cn}
\quad
\email{lintao@cuhk.edu.cn}
\quad
\email{jianweihuang@cuhk.edu.cn}
}

\maketitle

\begin{abstract}
\normalsize
AI advice increasingly enters human decision workflows before costly information acquisition and final action.
This paper studies how such advice affects decision quality, measured by the gross utility of the final action.
Advice shifts beliefs before the decision maker chooses an acquisition experiment with uniformly posterior-separable (UPS) costs.
The analysis characterizes gross effects through the continuation gross value generated by net-optimal acquisition.
Convexity of this value function is necessary and sufficient for every AI signal to weakly improve gross utility at every prior, while nonconvexity allows a binary signal to cause gross harm.
The criterion gives tight state-space results.
Every AI signal is weakly gross-improving in binary-state problems under arbitrary actions, payoffs, and UPS costs.
A minimal three-state problem can generate gross harm through concentrated shutdown, where valuable acquisition stops.
For classification problems with Shannon acquisition costs, the criterion becomes an active-set slope test.
It guarantees weak gross improvement for homogeneous classification and all three-state weighted classification problems, while a minimal four-state failure arises from attention dilution, where acquisition shifts away from high-stake distinctions.
Along a Blackwell-increasing mixture path, gross utility can first decrease and then increase.
With endogenous interpretation, cheaper interpretation can reduce gross utility.

\keywords{Rational Inattention \and AI-Assisted Decision-Making \and Endogenous Information Acquisition}
\end{abstract}

\section{Introduction}
\label{sec:intro}

AI systems are increasingly embedded in decision workflows and provide advice to human decision makers.
In many settings, people can gather additional information after receiving AI advice.
A clinician may observe an algorithmic risk score and then review additional records or order tests before choosing a treatment.
A worker may receive an answer from a generative assistant and then search for additional evidence or verify the answer before submitting a response~\cite{GreenChen2019,DeArteagaEtAl2020,DellAcquaEtAl2026,AgarwalMoehringWolitzky2025}.
Similar timing arises in lending, eligibility review, and other settings where AI advice is followed by human information gathering and action.

Such follow-up information acquisition is costly.
Reviewing records, running tests, searching, and verifying information require time, effort, attention, or other resources.
AI advice can therefore affect both what the decision maker knows and what additional information she chooses to acquire.
These effects jointly shape the final action, so whether AI advice improves decision quality is unclear.
This paper asks

\begin{center}
\textit{When is AI advice helpful or harmful in decision problems with costly information acquisition?}
\end{center}

The formal model is a Bayesian decision problem with active information acquisition.
A payoff-relevant state (unknown to the decision maker) is drawn from a known prior, and AI advice is modeled as an exogenous signal that changes the decision maker's belief about the state.
After receiving the AI advice, the decision maker chooses what additional information to acquire before acting.
Information acquisition costs are uniformly posterior-separable (UPS), so the cost of an experiment is measured by its expected reduction in a concave uncertainty measure.
The decision maker chooses acquisition strategy to maximize her net utility, which is the decision quality (expected payoff of the final action) minus the information acquisition cost.
While information acquisition is costly to the decision maker, we (as an analyst) focus on the gross utility of the decision maker, which does not include the information acquisition cost.
In many practical scenarios, decision quality is evaluated separately from the effort devoted to acquire information, and such effort is often privately borne by the decision maker and difficult to observe by the analyst. 
Our work thus use the gross utility as a measure to study when AI advice helps or harms decision making. 

We make three main contributions.
First, we characterize when AI advice improves or harms the gross utility of the decision maker through the posterior geometry of the acquisition problem.
The key object is the continuation gross value function.
At each belief, this function gives the gross utility generated by the optimized acquisition strategy.
We show that the necessary and sufficient condition for every AI signal to weakly improve gross utility at every prior is the convexity of the gross value function. 
When the function is nonconvex, a binary AI signal can change the optimal acquisition response in a way that strictly lowers expected gross utility.
This criterion reduces the evaluation of AI advice to a geometric property of the continuation problem.

Second, we derive general implications for decision problems.
In binary-state problems, every AI signal is weakly gross-improving for arbitrary action sets, payoffs, and UPS acquisition costs.
This binary-state result is tight in state-space size.
With three states, gross harm can arise through concentrated shutdown, where advice-induced posteriors are decisive enough to stop follow-up acquisition even though acquisition at the no-AI prior would have generated information valuable for the final action.

Third, we apply the geometric criterion to classification problems with Shannon acquisition costs and derive an active-set slope test.
The test shows that homogeneous classification and all three-state weighted classification problems are weakly gross-improving.
It also identifies a minimal weighted-classification failure with four states, where attention dilution shifts acquisition away from high-stake distinctions toward a lower-stake active class.

We also study how gross utility varies with the informativeness and interpretation of AI advice.
Whenever a gross-harmful signal exists, there is a Blackwell-increasing mixture path along which gross utility first decreases and then increases.
Extending the model to endogenous interpretation shows that cheaper interpretation can also reduce gross utility.

\subsection{Related Work}

This work connects AI-assisted human decisions, rational inattention, sequential learning, and information design.
We focus on AI advice that arrives before costly self-acquisition and ask when the resulting acquisition response improves or harms gross utility.

\paragraph{AI-Assisted Human Decisions.}
Recent work studies when AI assistance helps or hurts human decisions.
One line studies behavioral and informational frictions in the use of AI input, including reliance behavior~\cite{BucincaEtAl2021,DeJongEtAl2025,DietvorstSimmonsMassey2015,LoggMinsonMoore2019} and correlation neglect under overlapping human and AI information~\cite{AgarwalMoehringRajpurkarSalz2023,amin2026bayesian}.
Another line studies how confidence, feature transparency, and information-value diagnostics shape the usefulness of AI input~\cite{BalakrishnanFerreiraTong2024,LiSteyvers2025,GuoWuHartlineHullman2025}.
These papers examine how AI input is weighted, trusted, interpreted, or presented.
Related work also shows that informativeness alone can be an incomplete measure of algorithmic assistance when human learning or delegation responds endogenously~\cite{NotiEtAl2025,Xu2024AlgorithmAssisted}.
The closest paper is by Boyac{\i}, Canyakmaz, and de V\'ericourt~\cite{BoyaciCanyakmazDeVericourt2024}, who show that machine predictions weakly improve decision quality in a binary-state, binary-action homogeneous-classification model with Shannon information costs.
We show that this improvement result does not extend automatically.
In general decision problems with UPS acquisition costs, AI advice can be gross-harmful.
We characterize conditions for both gross improvement and gross harm, and identify their binary homogeneous-classification result as a gross-improving special case.

\paragraph{Rational Inattention.}
Rational inattention begins with Sims's model of costly information processing~\cite{Sims2003}, with Shannon entropy providing the canonical information measure~\cite{Shannon1948}.
The literature then developed tractable tools for static choice.
Mat\v{e}jka and McKay~\cite{MatejkaMcKay2015} derive the logit representation used in our Shannon calculations, while revealed-preference and active-set methods identify optimal attention from observed choices~\cite{CaplinDean2015,CaplinDeanLeahy2019,LipnowskiRavid2022}.
A parallel line studies alternative cost structures, including posterior-separable costs~\cite{Denti2022}, an axiomatic framework with constant marginal information costs~\cite{PomattoStrackTamuz2023}, and divergence-based costs~\cite{BloedelDentiPomatto2025}.
Another development studies costly information acquisition over time, including inertia and delay~\cite{SteinerStewartMatejka2017}, decisions that take time~\cite{HebertWoodford2023}, reversible and irreversible decisions~\cite{XuZariphopoulouZhang2023}, learning and stopping under ambiguity~\cite{AusterCheMierendorff2024}, and the costs of sequential information acquisition~\cite{BloedelZhong2025}.
We use these tools to study how exogenous AI advice changes acquisition and how the induced net-optimal policy affects the gross utility of the final action.

\paragraph{Sequential Learning.}
Sequential testing is the classical foundation for costly learning over time~\cite{Wald1947}, and social learning models show how early information can redirect later decisions and information flows~\cite{Banerjee1992,BikhchandaniHirshleiferWelch1992}.
More recent work studies how externally supplied or strategically provided information interacts with later information acquisition.
Matyskov\'a and Montes~\cite{MatyskovaMontes2023} study a Bayesian-persuasion problem in which a sender designs information while anticipating the receiver's costly acquisition.
Their focus is the sender's strategic design problem and the induced equilibrium information structure.
We instead take AI advice as exogenous and study how it changes the receiver's acquisition and the gross utility of the final action.
Cheng~\cite{Cheng2023} studies externally provided information is freely received before subsequent learning and characterizes how such a signal can affect the feasibility and direction of later Blackwell learning.
Our model shares the timing in which external information precedes endogenous learning, while our outcome variable is the gross utility selected by a net-optimal decision maker.

\paragraph{Information Design and Rationally Inattentive Receivers.}
Bayesian persuasion uses posterior splits as the central geometric object for designing information structures~\cite{KamenicaGentzkow2011}, and subsequent work extends this posterior-geometry perspective to broader design environments~\cite{BergemannMorris2019,Kamenica2019}.
A related rational-inattention literature studies receivers who process disclosed information under attention costs~\cite{BloedelSegal2018,LipnowskiMathevetWei2020,CheKimMierendorff2020}.
These papers ask how information should be designed or how strategically disclosed information is processed by a receiver facing attention costs.
We use the same posterior-split geometry to study how exogenous AI advice changes later acquisition and how the resulting net-optimal policy affects the gross utility of the final action.
Section~\ref{sec:discussion-extension} extends the baseline model to costly endogenous interpretation, which connects the model back to rationally inattentive processing of disclosed information.

\section{Model}
\label{sec:model}

This section sets up a model of AI advice before costly information acquisition.
Section~\ref{subsec:preliminaries} defines the classical problem of decision making with costly information acquisition.
Section~\ref{subsec:direct-ai-model} introduces AI advice as an exogenous signal before costly information acquisition.
Section~\ref{subsec:no-ai-benchmark} defines the no-AI benchmark and the gross-comparison notions used in the results.

\subsection{Preliminaries: Decision-Making with Costly Information Acquisition}
\label{subsec:preliminaries}

Let \(X\in\X\) be a payoff-relevant state with prior \(\mu\in\DeltaX\), where \(\X\) is finite.
There is a finite set of actions $\A$.
The payoff of action $a \in \A$ in state \(x \in \X\) is \(u(a,x)\).
For any belief \(\nu\in\DeltaX\), the static gross value is
\begin{equation}
    v(\nu)
    =
    \max_{a\in\A}\sum \nolimits_{x\in\X}\nu(x) u(a,x).
\label{eq:static-gross-value}
\end{equation}
The function \(v : \Delta(\X) \to \mathbb{R} \) is convex because it is the maximum of finitely many affine functions.

Before taking an action, the decision maker can conduct an experiment to acquire information.
A statistical experiment is a signal channel $X \to S$ about the payoff-relevant state.
Given a starting belief $\mu$, each signal realization $S = s$ induces a posterior belief $\mu(X|s)$ by Bayes' rule.
Because both payoffs and posterior-separable costs depend on the experiment only through the distribution of the induced posteriors, we use posterior splits as the primitive representation.

\begin{definition}[Posterior split]
\label{def:posterior-split}
For a prior \(\mu\in\DeltaX\), a \emph{posterior split} of \(\mu\) is a distribution over beliefs in \(\DeltaX\) with mean \(\mu\).
The set of all (finite) posterior splits is denoted by
\begin{equation}
\mathcal E_\X(\mu)=
\left\{
\sum \nolimits_{i\in\mathcal I}\alpha_i\delta_{\nu^i}
~\Big|~
|\mathcal I|<\infty,\
\alpha_i>0,\
\sum \nolimits_{i\in\mathcal I} \alpha_i=1,\
\nu^i\in\DeltaX,\
\sum\nolimits_{i\in \mathcal I}\alpha_i\nu^i=\mu
\right\}.
\label{eq:posterior-split-definition}
\end{equation}
Here \(\delta_{\nu^i}\) is the Dirac distribution at belief \(\nu^i\).
The mean condition is Bayes plausibility.
Zero-weight terms are omitted throughout.
We identify each experiment about \(X\) with the posterior split it induces, and every element of \(\mathcal E_\X(\mu)\) is implementable by such an experiment.
\end{definition}

The information acquisition cost is the reduction in uncertainty from the starting belief to the posterior beliefs generated by the chosen posterior split~\cite{CaplinDean2015,Denti2022}.

\begin{definition}[UPS Acquisition Cost]
\label{def:ups-cost}
An \emph{uncertainty function} is a continuous concave function \(\Phi:\DeltaX\to\R\).
The \emph{uniformly posterior-separable (UPS)} acquisition cost generated by \(\Phi\) is
\begin{equation}
    C^\Phi(\mu;\eta)
    =
    \Phi(\mu)-\E_{\nu\sim\eta}\Phi(\nu),
    \qquad
    \eta\in\mathcal E_\X(\mu).
\label{eq:ups-acquisition-cost}
\end{equation}
\end{definition}

The cost in Eq.~\eqref{eq:ups-acquisition-cost} is nonnegative by concavity of \(\Phi\).
The canonical uncertainty function is Shannon entropy,
\[
    H_\X(\mu)
    =
    -\sum \nolimits_{x\in\X}\mu(x)\log \mu(x),
    \qquad 0\log0:=0.
\]
When the state space is clear, we write \(H\).
For any experiment implementing posterior split \(\eta\), the Shannon cost \(C^{H_\X}(\mu;\eta)\) equals the mutual information between the state and the experiment signal.

We use concavification to represent the value of posterior-split problems.

\begin{definition}[Concavification]
\label{def:concavification}
For a continuous function \(f:\DeltaX\to\mathbb R\), its concavification is \(\cav f(\mu)=\sup_{\eta\in\mathcal E_\X(\mu)}\E_{\nu\sim\eta}f(\nu)\).
Equivalently, \(\cav f\) is the smallest concave majorant of \(f\) on \(\DeltaX\).
\end{definition}

The decision maker chooses a posterior split to maximize expected gross utility minus acquisition cost.
Under UPS costs, this problem reduces to concavification in posterior-split geometry~\cite{KamenicaGentzkow2011}.

\begin{lemma}[Concavification of Acquisition]
\label{lem:acquisition-cav}
For any belief \(\mu\in\DeltaX\), the continuation net value is
\begin{equation}
    W_\Phi(\mu)
    =
    \sup_{\eta\in\mathcal E_\X(\mu)}
    \left\{\E_{\nu\sim\eta}v(\nu)-C^\Phi(\mu;\eta)\right\}
    =
    \cav(v+\Phi)(\mu)-\Phi(\mu).
\label{eq:continuation-net-value}
\end{equation}
Moreover, a net-optimal acquisition split exists.
\end{lemma}

The proof is in Appendix~\ref{app:proof-acquisition}.
The lemma represents costly acquisition as the choice of a Bayes-plausible posterior split and identifies the optimal split through concavification.
The optimization problem in Eq.~\eqref{eq:continuation-net-value} can have multiple net-optimal posterior splits with different gross utilities.
Among net-optimal splits, we select one with the largest gross utility.
We denote this selected continuation gross value and its acquisition cost by \(G_\Phi(\mu)\) and \(K_\Phi(\mu)\), respectively.
The gross-favorable selection rule is conservative: if gross utility falls under this convention, the loss is not an artifact of pessimistic tie-breaking among net-optimal acquisition policies.

\subsection{AI-Assisted Decision-Making}
\label{subsec:direct-ai-model}

Now fix the original prior \(\mu_0\in\DeltaX\).
We formalize AI advice as an external signal \(Y\in\Y\) generated from a known conditional distribution \(P_{Y|X}\).
Before taking action, the decision maker observes \(Y\).
The AI signal induces the posterior split \( \eta^Y = \sum_{y\in \Y} \alpha_y\delta_{\mu_y} \in\mathcal E_\X(\mu_0), \) where \(\alpha_y=P(Y=y)\) is the marginal probability of realization \(y\) and \(\mu_y=P(X\mid Y=y)\) is the posterior belief after that realization.

Conditioning on realization \(y\), the decision maker starts from belief \(\mu_y\), acquires more information by choosing a posterior split in \(\mathcal E_\X(\mu_y)\) at a cost, and then acts optimally at each acquisition-induced posterior belief.
Averaging over AI signal realizations gives the AI-assisted gross utility, net value, and acquisition cost,
\[
G^{\mathrm{AI}}_\Phi(Y,\mu_0)
=
\E_{y\sim Y}G_\Phi(\mu_y),\qquad
W^{\mathrm{AI}}_\Phi(Y,\mu_0)
=
\E_{y\sim Y}W_\Phi(\mu_y),\qquad
K^{\mathrm{AI}}_\Phi(Y,\mu_0)
=
\E_{y\sim Y}K_\Phi(\mu_y).
\]
The gross-comparison results below use \(G^{\mathrm{AI}}_\Phi(Y,\mu_0)\) as the outcome variable.

\subsection{No-AI Benchmark and Gross Comparisons}
\label{subsec:no-ai-benchmark}

The no-AI benchmark is the acquisition problem starting from the original prior $\mu_0$.
Its gross value is \(G_\Phi(\mu_0)\).
For gross comparisons, define
\begin{equation}
    \Delta G_\Phi(Y,\mu_0)
    =
    G^{\mathrm{AI}}_\Phi(Y,\mu_0)-G_\Phi(\mu_0).
\label{eq:delta-gross}
\end{equation}

\begin{definition}[Gross-Improving and Gross-Harmful AI Signals]
\label{def:gross-improvement-harm}
Fix a payoff function and an uncertainty function \(\Phi\).
Given a prior \(\mu_0\), an AI signal \(Y\) is \emph{weakly gross-improving} at \(\mu_0\) if \(\Delta G_\Phi(Y,\mu_0)\ge 0\), and is \emph{gross-harmful} at \(\mu_0\) if \(\Delta G_\Phi(Y,\mu_0)<0\).
\end{definition}

The rest of the paper characterizes when AI advice is weakly gross-improving and when it generates gross harm.

\section{Characterization for General Decision Problems}
\label{sec:mechanism}
\label{sec:state-geometry}

We begin with the characterization for general decision problems.
The argument proceeds from the general criterion to its state-space implications.
Section~\ref{subsec:gross-criterion} gives the criterion.
Section~\ref{subsec:binary-gross-improvement} applies it to binary-state problems, where one-dimensional belief geometry rules out gross harm.
Section~\ref{subsec:treatment} gives a decision problem with three states where an AI signal can be gross-harmful.

\subsection{Posterior-Geometry Criterion}
\label{subsec:gross-criterion}

AI advice induces a Bayes-plausible posterior split of the prior.
We show that its effect on gross utility is the Jensen gap of \(G_\Phi\) under that split.

\begin{theorem}[Posterior Geometry of AI Advice]
\label{thm:posterior-geometry}
Fix a payoff function and a UPS acquisition cost function \(C^\Phi\).
\begin{enumerate}[label=\textup{(\roman*)}, leftmargin=1.8em]
\item If the continuation gross value function \(G_\Phi\) is convex on \(\DeltaX\), every AI signal is weakly gross-improving at every prior.
\item If the continuation gross value function \(G_\Phi\) is not convex on \(\DeltaX\), there exist a prior and a binary AI signal that generates gross harm.
\end{enumerate}
\end{theorem}

The proof is in Appendix~\ref{app:proof-posterior-geometry}.
Bayes plausibility gives \(\sum_y\alpha_y\mu_y=\mu_0\).
If \(G_\Phi\) is convex, Jensen's inequality makes the gross-utility change nonnegative for every posterior split.
If \(G_\Phi\) is nonconvex, a strict Jensen violation gives two posterior beliefs whose average is a prior, and that pair of posteriors can be induced by a binary AI signal.\footnote{The binary signal in part~\textup{(ii)} is only a minimal witness: any AI signal generates gross harm when its induced posterior split has a negative Jensen gap for \(G_\Phi\).}
Economically, the theorem reduces the effect of AI advice to the geometry of continuation gross value.
Convexity means that moving the decision maker to AI-induced posterior branches cannot lower the average gross utility selected by net-optimal acquisition, while a strict Jensen loss identifies a split at which AI advice sends the decision maker to branches whose optimized acquisition policies deliver less gross utility than the policy selected at the average belief.
We then apply this criterion to state-space geometry.

\subsection{Binary State: Always Weakly Gross-Improving}
\label{subsec:binary-gross-improvement}

Binary-state decision problems are geometrically special because all posterior beliefs lie on a line.
We show that this one-dimensional geometry is enough to make every AI signal weakly gross-improving, regardless of the payoff function or the UPS acquisition cost.

\begin{figure}[t]
\centering
\includegraphics[width=.82\linewidth]{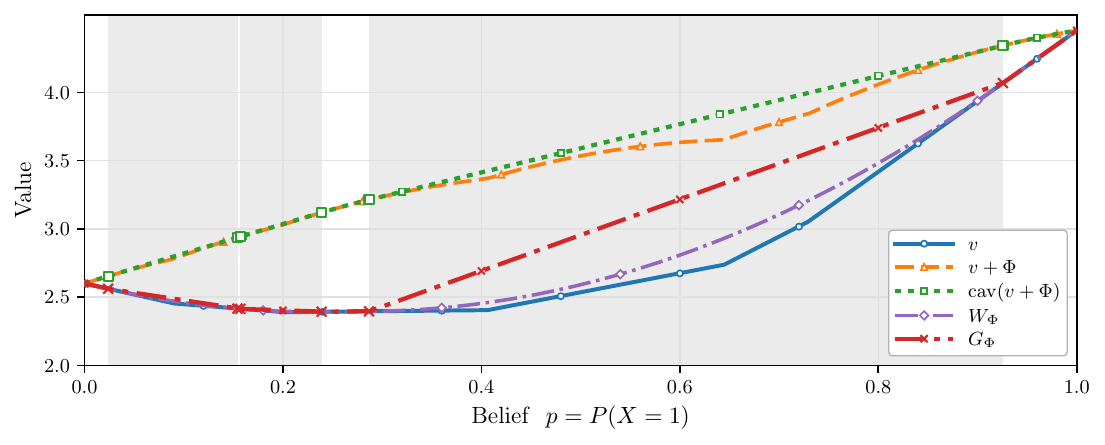}
\caption{Binary-state chord geometry with six actions and uncertainty function \(\Phi(p)=4p(1-p)\).
The shaded gray region marks beliefs where chord replacement occurs.}
\label{fig:binary-gini-gross-improvement}
\end{figure}

\begin{theorem}[Binary-State Problems Are Always Weakly Gross-Improving]
\label{thm:binary-ups}
Suppose \(|\X|=2\).
For any finite action set \(\A\), payoff function \(u:\A\times\X\to\mathbb R\), and UPS acquisition cost function \(C^\Phi\), every AI signal is weakly gross-improving at every prior.
\end{theorem}

Boyac{\i}, Canyakmaz, and de V\'ericourt~\cite{BoyaciCanyakmazDeVericourt2024} show that AI advice does not harm gross utility in specific binary-state, binary-action problems.
Our result extends this conclusion to arbitrary binary-state problems.
The proof is in Appendix~\ref{app:proof-binary-ups}.
Fig.~\ref{fig:binary-gini-gross-improvement} illustrates the argument.
Beliefs are indexed by \(p=P(X=1)\).
The net acquisition problem is equivalent to concavifying \(v+\Phi\) and then subtracting \(\Phi\).
On an affine segment of this concavification, a net-optimal acquisition policy mixes between the endpoint contact beliefs.
The selected gross value follows the chord of \(v\) between those endpoints rather than the local segment of \(v\).
The shaded gray region marks the beliefs where chord replacement occurs.
Replacing disjoint intervals of a convex function by endpoint chords preserves convexity.
The same secant-slope argument also covers the minimal-information tie-breaking rule, so the binary-state conclusion does not rely on the gross-favorable selection convention.

The intuition is that, with a binary state space, AI signals cannot redirect acquisition across distinct state comparisons.
An AI signal can make one state more likely or less likely, and it can change how much acquisition the decision maker buys, but all posterior movements remain on the same belief line.
Richer state spaces allow advice to redirect acquisition across different state dimensions.
The next example gives the minimal failure.

\subsection{A Minimal Failure: Three-State Treatment Example}
\label{subsec:treatment}

The binary-state gross-improvement guarantee is tight: with three states, AI advice can be gross-harmful.
When reporting numerical examples below, we use the same AI-minus-benchmark convention for the net-value change and the acquisition-cost change: \(\Delta W_\Phi(Y,\mu_0)=W^{\mathrm{AI}}_\Phi(Y,\mu_0)-W_\Phi(\mu_0)\) and \(\Delta K_\Phi(Y,\mu_0)=K^{\mathrm{AI}}_\Phi(Y,\mu_0)-K_\Phi(\mu_0)\).

\begin{wineexample}[Three-State Treatment Failure]
\label{ex:treatment}
Let \(\X=\{L,M,H\}\) and \(\A=\{0,1\}\), where action \(1\) is treatment and action \(0\) is no treatment.
The prior is
\[
    \mu_0=(0.3,0.6,0.1).
\]
Payoffs are
\[
    u(0,x)=0\quad\text{for all }x,
    \qquad
    u(1,L)=1,
    \quad
    u(1,M)=-3,
    \quad
    u(1,H)=3.
\]
Treatment is valuable in the severe state \(H\), harmful in the contraindicated state \(M\), and mildly valuable in the low-risk state \(L\).
The uncertainty function is \(\Phi=H_\X\).

\emph{No-AI benchmark.}
At the prior, the decision maker acquires diagnostic information before acting.
Both actions are active in the optimal Shannon action prior, with approximate action-prior vector \( (0.77230,0.22770) \) over actions \((0,1)\).
The no-AI continuation values are approximately
\[
    W_\Phi(\mu_0)=0.12041,
    \qquad
    G_\Phi(\mu_0)=0.36409,
    \qquad
    K_\Phi(\mu_0)=0.24368.
\]

\emph{AI-assisted policy.}
A binary AI signal \(Y\in\{h,m\}\) perfectly identifies the severe state:
\[
    P(h\mid H)=1,
    \qquad
    P(h\mid L)=P(h\mid M)=0.
\]
Hence \(P(Y=h)=0.1\), \(P(Y=m)=0.9\), and
\[
    P(X\mid Y=h)=(0,0,1),
    \qquad
    P(X\mid Y=m)=\left(\frac13,\frac23,0\right).
\]
After realization \(h\), the decision maker treats without further acquisition.
After realization \(m\), she does not treat and acquires no further information.
The AI-assisted continuation values are
\[
    W^{\mathrm{AI}}_\Phi(Y,\mu_0)
    =
    G^{\mathrm{AI}}_\Phi(Y,\mu_0)
    =
    0.30000,
    \qquad
    K^{\mathrm{AI}}_\Phi(Y,\mu_0)=0.
\]

\emph{Comparison.}
Relative to the no-AI benchmark,
\[
    \Delta W_\Phi(Y,\mu_0)=0.17959,
    \qquad
    \Delta G_\Phi(Y,\mu_0)=-0.06409,
    \qquad
    \Delta K_\Phi(Y,\mu_0)=-0.24368.
\]
AI advice raises net value while reducing gross utility by approximately \(17.6\%\) relative to the no-AI benchmark.
\end{wineexample}

The failure mechanism is \emph{concentrated shutdown}.
The prior supports valuable diagnostic acquisition, while both advice-induced posteriors lead the decision maker to stop acquiring information.
This mechanism is related to empirical evidence that algorithmic recommendations and explanations can change reliance without improving gross utility~\cite{GreenChen2019,BucincaEtAl2021,PoursabziSangdehEtAl2021}.
The practical implication is to audit coarse recommendations, confidence flags, and explanation summaries for whether they suppress follow-up testing, record review, or verification in cases where such acquisition is decision-relevant.
Appendix~\ref{app:treatment-verification} gives the numerical verification.
The gross-loss margin and no-acquisition optimality conditions are strict.
Together with the local support argument in Appendix~\ref{app:shannon-kkt}, these inequalities imply that the acquisition pattern and loss sign persist under sufficiently small feasible perturbations.

\paragraph{Net utility.}
The net gain in Example~\ref{ex:treatment} reflects a property of Shannon acquisition costs.
Under these costs, Corollary~\ref{cor:shannon-net-benefit} shows that every AI signal weakly raises net utility.
This conclusion need not hold for other uncertainty functions.
Appendix~\ref{app:gini-treatment} gives a Quadratic/Gini example in which AI advice lowers both gross and net utility.

\section{Application: Classification Problems with Shannon Acquisition Costs}
\label{sec:payoff-geometry}

We now apply the general posterior-geometry criterion to classification problems, an economically important family of decision problems that includes medical triage, loan review, hiring screens, content moderation, and routing cases to service queues.
Section~\ref{subsec:classification-slope} states the classification criterion as an active-set slope test.
Section~\ref{subsec:classification-gross-improvement} identifies weakly gross-improving classification cases.
Section~\ref{subsec:classification-failure} gives a gross-harmful weighted-classification example.

\subsection{Classification Gross-Comparison Criterion}
\label{subsec:classification-slope}

Let \(\X=\A=\{1,\ldots,n\}\).
The action is the predicted label and the state is the true label.
Consider the classification payoff
\begin{equation}
    u(a,x)
    =
    b_x+w_x\one\{a=x\},
    \qquad
    w_x>0,
\label{eq:classification-payoff}
\end{equation}
where \(b_x\) is the state baseline payoff and \(w_x\) is the correct-label reward.
In the homogeneous case, all weights are equal.

Under Shannon acquisition with uncertainty function \(\lambda H_\X\), the selected continuation gross value is piecewise affine in beliefs.
For a belief, call a label \emph{active} if it is chosen with positive probability in the optimal Shannon continuation problem.
An active set is the collection of labels that receive acquisition attention at that belief.
The active-set cell for a subset \(J\) is the region of beliefs whose active labels are exactly \(J\).
When a belief crosses from the cell \(J\) to the neighboring cell \(J\cup\{k\}\), label \(k\) is the entering label.
The slope test below translates the Jensen-gap criterion from Theorem~\ref{thm:posterior-geometry} into a boundary condition: an entrant must not lower the gross slope by diluting attention away from labels already active.

\begin{theorem}[Classification Slope Test]
\label{thm:slope-criterion}
Consider the classification payoff in Eq.~\eqref{eq:classification-payoff} under acquisition with Shannon uncertainty function \(\lambda H_\X\), where \(\lambda>0\).
For each label \(i\), set \(E_i=\exp(w_i/\lambda)\).
For every nonempty active set \(J\subseteq\X\) and every label \(k\notin J\), define the adjacent-boundary slope jump \(\Delta_{J,k}\) by
\begin{equation}
\Delta_{J,k}
=
\frac{w_k E_k}{E_k-1}
-
\frac{
\sum_{i\in J}{w_i E_i}/{(E_i-1)^2}
}{
1+\sum_{i\in J}{1}/{(E_i-1)}
}.
\label{eq:slope-expression}
\end{equation}
A boundary is feasible if the active-set cells for \(J\) and \(J\cup\{k\}\) share a nonempty facet.
If every feasible adjacent boundary satisfies \(\Delta_{J,k}\ge0\), every AI signal is weakly gross-improving.
If some feasible adjacent boundary satisfies \(\Delta_{J,k}<0\), the environment admits a gross-harmful AI signal.
\end{theorem}

The proof is in Appendix~\ref{app:proof-slope-criterion}.
In Eq.~\eqref{eq:slope-expression}, the first term is the entering label's gross marginal contribution.
The second term is the active block's dilution term.
A negative slope jump means that the entrant bends the selected continuation gross value downward as it joins the active set.
High-stake entrants tend to pass the test.
Low-stake labels can fail it by entering an active block and diluting acquisition away from more valuable labels.
The active-set formulas agree on shared feasible facets, so the slope test does not depend on tie-breaking among net-optimal action priors.

\subsection{Weakly Gross-Improving Classification}
\label{subsec:classification-gross-improvement}

The first implication is that equal stakes preserve convexity of the selected continuation gross value.

\begin{proposition}[Homogeneous Classification Is Weakly Gross-Improving]
\label{prop:homogeneous-classification}
For the homogeneous specialization of Eq.~\eqref{eq:classification-payoff}, with \(w_x\equiv w>0\), and Shannon acquisition with uncertainty function \(\lambda H_\X\), where \(\lambda>0\), the continuation gross value is
\begin{equation}
G_{\lambda H_\X}(\mu)
=
\sum_x\mu(x) b_x
+
\max_{\varnothing\ne J\subseteq\X}
\frac{w e^{w/\lambda}}{e^{w/\lambda}-1+|J|}
\sum_{x\in J}\mu(x).
\label{eq:hom-class-g}
\end{equation}
Consequently, every AI signal is weakly gross-improving at every prior.
\end{proposition}

The proof is in Appendix~\ref{app:proof-homogeneous-classification}.
The slope test gives the mechanism.
When all labels have the same stake, every feasible adjacent-boundary slope jump is nonnegative: adding an active label cannot dilute acquisition away from a more important label because no label is more important.
This explains why the homogeneous binary classification environment in Boyac{\i}, Canyakmaz, and de V\'ericourt~\cite{BoyaciCanyakmazDeVericourt2024} ensures decision-quality improvement.
Eq.~\eqref{eq:hom-class-g} determines \(G_{\lambda H_\X}\) directly, so the conclusion does not rely on the gross-favorable selection rule.

Beyond homogeneous stakes, the slope test still rules out gross harm in three-label classification.

\begin{proposition}[Three-State Weighted Classification Is Weakly Gross-Improving]
\label{prop:three-state-gross-improvement}
With \(|\X|=|\A|=3\) and Shannon acquisition, every AI signal is weakly gross-improving at every prior for every classification payoff in Eq.~\eqref{eq:classification-payoff}.
\end{proposition}

Appendix~\ref{app:proof-three-state-classification} proves the proposition by applying the slope test to all feasible adjacent boundaries.
With three labels, an entrant joins an active set containing at most two labels.
In either case, the active-block dilution term is below the acquisition multiplier, whereas the entrant's gross contribution is above it.
All feasible adjacent-boundary slope jumps are positive.
The first weighted-classification failure requires at least four labels, where such dilution can occur.

\subsection{Weighted Classification Gross Harm Example}
\label{subsec:classification-failure}

Weighted classification can fail through attention dilution.
A low-stake category can become salient enough after advice to enter the active set, and its entry can reduce acquisition devoted to higher-stake distinctions.
The next example adds one low-stake routine category to three high-stake categories and verifies a strict Jensen loss.

\begin{wineexample}[Four-Class Triage Dilution]
\label{ex:weighted-four-state}
Consider a triage classifier with three high-consequence categories and one lower-consequence routine category.
Let \(\X=\A=\{1,2,3,4\}\).
The uncertainty function is \(\Phi=100H_\X\).
Payoffs are
\[
    u(a,x)=w_x\one\{a=x\},
    \qquad
    (w_1,w_2,w_3,w_4)=(30,30,30,5).
\]
Correct classification carries high stakes in the first three categories.
The fourth is a lower-value category, such as a routine follow-up class, a low-priority queue, or a benign category.
The prior is
\[
    \mu_0=
    \left({19}/{150},{19}/{150},{19}/{150},{31}/{50}\right).
\]

\emph{No-AI benchmark.}
At the prior, the no-AI continuation values are approximately
\[
    W_\Phi(\mu_0)=4.19162,
    \qquad
    G_\Phi(\mu_0)=4.59374,
    \qquad
    K_\Phi(\mu_0)=0.40212.
\]
The optimal Shannon active set is \(\{1,2,3\}\), with action-prior vector \((1/3,1/3,1/3,0)\).
Thus the no-AI benchmark concentrates acquisition on the three high-stake classes and leaves the low-stake routine class inactive.

\emph{AI-assisted policy.}
A binary AI signal \(Y\in\{\ell,r\}\) has probabilities \(P(Y=\ell)=4/5\) and \(P(Y=r)=1/5\), with posterior beliefs
\[
\mu^\ell=
\left({1}/{10},{1}/{10},{1}/{10},{7}/{10}\right),
\qquad
\mu^r=
\left({7}/{30},{7}/{30},{7}/{30},{3}/{10}\right).
\]
Their probability-weighted average is the prior.
On the routine-heavy branch \(\ell\), all four classes are active and the action prior shifts heavily toward the low-stake routine class.
On branch \(r\), the active set remains \(\{1,2,3\}\).
The AI-assisted values are approximately
\[
    W^{\mathrm{AI}}_\Phi(Y,\mu_0)=4.34931,
    \qquad
    G^{\mathrm{AI}}_\Phi(Y,\mu_0)=4.56707,
    \qquad
    K^{\mathrm{AI}}_\Phi(Y,\mu_0)=0.21776.
\]

\emph{Comparison.}
Relative to the no-AI benchmark,
\[
    \Delta W_\Phi(Y,\mu_0)=0.15768,
    \qquad
    \Delta G_\Phi(Y,\mu_0)=-0.02668,
    \qquad
    \Delta K_\Phi(Y,\mu_0)=-0.18436.
\]
Advice strictly raises net value and generates gross harm.
It saves acquisition cost by shifting a high-probability posterior toward a low-stake routine category, while the induced attention dilution lowers gross classification utility.
\end{wineexample}

Appendix~\ref{app:triage-verification} gives the numerical verification.
The failure mechanism is attention dilution.
A low-consequence routine category enters the active set and weakens acquisition for consequential categories.
The endogenous acquisition response is related to evidence
that AI predictions can crowd out human effort
\cite{AgarwalMoehringWolitzky2025}.
The allocation of acquisition across distinctions has a conceptual connection to multitask incentive problems~\cite{HolmstromMilgrom1991,Baker1992}.
The practical guidance is to audit classwise accuracy together with how advice changes follow-up review across high-stake and routine categories.
The strict loss and optimality conditions in Appendix~\ref{app:triage-verification} also establish local robustness by the support argument in Appendix~\ref{app:shannon-kkt}.

\section{Discussion and Extension}
\label{sec:discussion-extension}

The previous sections characterize weak gross improvement and gross harm for fixed, exogenous AI advice.
This section records two implications that matter for design.
Section~\ref{subsec:informativeness} studies comparative statics in the informativeness of AI advice.
Section~\ref{subsec:costly-interpretation-extension} extends to endogenous interpretation of AI advice and then studies comparative statics in interpretation cost.

\begin{figure}[t]
\centering
\begin{minipage}[t]{.48\linewidth}
\centering
\includegraphics[width=\linewidth]{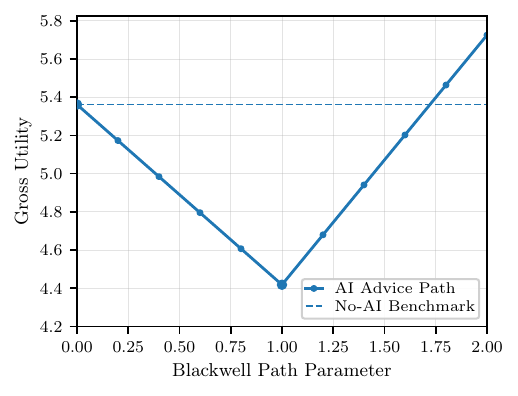}
\caption{Gross utility along a Blackwell-increasing AI advice path in the safe-option environment.
Settings: \(\X=\A=\{1,2,3\}\), action \(2\) is safe, and the payoff matrix has rows \((6,3,-12)\), \((1,3,1)\), and \((-12,3,6)\), with states in rows and actions in columns.
The uncertainty function is \(\Phi=4H_\X\).}
\label{fig:informativeness-gross}
\end{minipage}\hfill
\begin{minipage}[t]{.48\linewidth}
\centering
\includegraphics[width=\linewidth]{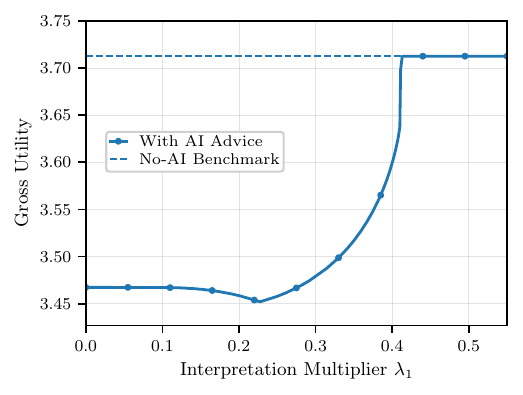}
\caption{Gross utility with endogenous interpretation and Shannon costs in the distance-payoff environment.
Settings: \(\X=\A=\{1,2,3\}\), \(u(a,x)=4-(x-a)^2\), \(\Phi=1.5H_\X\), \(\Psi=H_\Y\), and interpretation uncertainty function \(\lambda_1\Psi\).}
\label{fig:costly-interpretation-gross}
\end{minipage}
\end{figure}

\subsection{Signal Informativeness}
\label{subsec:informativeness}

More informative advice is usually expected to help because it gives the decision maker a finer posterior belief before action.
However, gross utility can fall, as the next result shows.

\begin{proposition}[Gross Utility First Decreases and Then Increases]
\label{prop:informativeness-nonmonotonicity}
Fix a payoff function, a UPS uncertainty function \(\Phi\), and a prior \(\mu_0\).
If a gross-harmful AI signal \(Y\) exists at \(\mu_0\), then there exists a continuous family of AI signals, strictly increasing in the Blackwell order, along which gross utility first strictly decreases and then strictly increases.
\end{proposition}

The construction follows two mixture paths.
Starting from no information, it reveals the harmful signal with increasing probability, so gross utility moves linearly toward the lower value generated by that signal.
It then replaces the harmful signal with full revelation with increasing probability, so gross utility moves linearly toward the statewise optimal value.
Each step increases informativeness in the Blackwell order.
Appendix~\ref{app:proof-informativeness-nonmonotonicity} gives the formal construction and proof.
Fig.~\ref{fig:informativeness-gross} illustrates this two-segment path in the safe-option environment.

This result complements other work showing that informativeness alone need not determine the value of algorithmic assistance.
More informative features can trade off short-run prediction accuracy against human learning~\cite{NotiEtAl2025}, while maximal informativeness can be suboptimal when delegation interacts with private information and preference misalignment~\cite{Xu2024AlgorithmAssisted}.
Here the channel is subsequent costly acquisition.
More informative AI advice changes the belief from which acquisition is chosen, so gross utility need not increase monotonically with informativeness.
For practical design, improving informativeness alone is insufficient.
One must also evaluate how the AI signal changes follow-up acquisition and how the induced final action affects gross utility.

\subsection{Extension: Endogenous and Costly Interpretation}
\label{subsec:costly-interpretation-extension}

In the baseline model, the decision maker observes the AI signal directly.
A natural extension lets the decision maker endogenously interpret AI advice before acquisition.
Let \(\bar q=P_Y\in\DeltaY\) be the marginal distribution of the AI signal, and let \(\Psi:\DeltaY\to\R\) be a concave uncertainty function for interpretation.
For a belief \(q\in\DeltaY\) over AI signal realizations, define the Bayes map
\begin{equation}
    \rho(q)
    =
    \sum_{y\in\Y}q_y\mu_y
    \in\DeltaX,
\label{eq:extension-bayes-map}
\end{equation}
where \(\mu_y\) is the posterior belief after AI signal \(y\).
An interpretation policy is a posterior split \(\zeta\in\mathcal E_\Y(\bar q)\).
With interpretation multiplier \(\lambda_1\ge0\), the optimal value is
\begin{equation}
W^{\mathrm{AI}}_{\lambda_1\Psi,\Phi}(Y,\mu_0)
=
\sup_{\zeta\in\mathcal E_\Y(\bar q)}
\left\{
    \E_{q\sim\zeta}W_\Phi(\rho(q))
    -
    C^{\lambda_1\Psi}(\bar q;\zeta)
\right\}
=
\cav\big(W_\Phi\circ\rho+\lambda_1\Psi\big)(\bar q)
-
\lambda_1\Psi(\bar q).
\label{eq:extension-value}
\end{equation}
The extended gross utility \(G^{\mathrm{AI}}_{\lambda_1\Psi,\Phi}(Y,\mu_0)\) is the largest expected continuation gross value among interpretation policies that attain Eq.~\eqref{eq:extension-value}.

The first question is whether endogenous interpretation eliminates gross harm.
The answer is no.
Even when interpretation is completely flexible and costless, the decision maker may interpret AI advice because it improves the continuation net objective by saving acquisition cost.
After interpretation, the induced reduction in self-acquisition can nevertheless lower the gross utility of the final action.
Shannon acquisition costs provide an implementation of this mechanism.

\begin{theorem}[Endogenous Interpretation Does Not Eliminate Gross Harm]
\label{thm:shannon-nonconvexity-harm}
Let \(\Phi=\lambda H_{\X}\) for some \(\lambda>0\).
If \(G_\Phi\) is not convex on \(\DeltaX\), then there exist a prior \(\mu_0\in\DeltaX\) and a binary AI signal \(Y\) such that full interpretation is the unique net-optimal interpretation posterior split at zero interpretation cost, and the interpreted AI signal generates strict gross harm: \( G^{\mathrm{AI}}_{0,\Phi}(Y,\mu_0) < G_\Phi(\mu_0). \)
\end{theorem}

Theorem~\ref{thm:shannon-nonconvexity-harm} implements the gross-harmful posterior split from Theorem~\ref{thm:posterior-geometry} in the endogenous interpretation problem.
Under Shannon costs, this split also gives a strict net gain, which makes full interpretation uniquely net-optimal at zero interpretation cost.
The conclusion of gross harm follows from the same posterior split because the displaced acquisition is valuable for the final action.
Appendix~\ref{app:proof-endogenous-interpretation} gives the proof.

Having established that endogenous interpretation does not eliminate gross harm, we next ask how interpretation cost changes the selected gross outcome.
A natural intuition is that easier-to-interpret AI advice should improve decision quality, because the decision maker can extract more information at a lower private cost.
The next result shows that this intuition can fail for gross utility.

\begin{proposition}[Cheaper Interpretation Can Lower Gross Utility]
\label{prop:cheaper-interpretation-harms}
There exist a payoff function, a prior \(\mu_0\), an AI signal \(Y\), an acquisition uncertainty function \(\Phi\), an interpretation uncertainty function \(\Psi\), and interpretation-cost multipliers \(0\le \lambda_1'<\lambda_1''\) such that
\[
    G^{\mathrm{AI}}_{\lambda_1'\Psi,\Phi}(Y,\mu_0)
    <
    G^{\mathrm{AI}}_{\lambda_1''\Psi,\Phi}(Y,\mu_0).
\]
\end{proposition}

Appendix~\ref{app:proof-cheaper-interpretation} proves the result by comparing any two-stage policy with direct acquisition of its joint signal.
With Shannon costs at both stages and \(\lambda_1>\lambda\), the data processing inequality makes direct acquisition strictly cheaper than any policy that extracts information from the AI signal before acquiring further information.
The optimal policy then leaves the AI signal uninterpreted and recovers the no-AI continuation benchmark, while Theorem~\ref{thm:shannon-nonconvexity-harm} gives strict gross harm at zero interpretation cost.
Fig.~\ref{fig:costly-interpretation-gross} illustrates a richer pattern at intermediate interpretation costs.
When interpretation is cheap, the interpreted advice can suppress follow-up acquisition and reduce gross utility.
As interpretation becomes more costly, less information is extracted from the advice, but follow-up acquisition may remain suppressed, leaving both information channels weakened.
Once the interpreted advice becomes too weak to discourage further acquisition, the decision maker resumes acquisition and gross utility returns toward the no-AI benchmark.
Appendix~\ref{app:distance-interpretation} records this numerical path.

\section{Conclusion}
\label{sec:conclusion}

The paper uses a rational inattention model to characterize when AI advice improves or harms final decision quality.
After receiving advice, a Bayesian decision maker updates her belief and then chooses costly acquisition to maximize net utility.
We evaluate the gross utility of the final action before acquisition costs are subtracted.
Advice can lower final decision quality by changing which acquisition policy is optimal for the net objective.

The results have three parts.
First, the continuation gross value function gives a general criterion.
Convexity is necessary and sufficient for every AI-induced posterior split to be weakly gross-improving.
Nonconvexity implies that some binary AI signal generates gross harm.
Second, the criterion gives a tight state-space boundary for general decision problems.
In every binary-state environment, every AI signal is weakly gross-improving under arbitrary action sets, payoffs, and UPS costs.
A minimal three-state example generates concentrated shutdown.
Third, applying the criterion to classification under Shannon acquisition costs gives an active-set slope test.
The test proves that homogeneous classification and all three-state weighted classification problems are weakly gross-improving, and it identifies a minimal four-state failure driven by attention dilution.

The practical implication is that AI evaluation should track acquisition behavior alongside final choices.
Deployment audits should record follow-up testing, record review, verification, and search, especially when these activities carry high gross value.
More informative AI advice is not a sufficient design target, since stronger belief movements can weaken valuable downstream acquisition before the AI signal is precise enough to replace it.
The endogenous-interpretation analysis gives the same warning for interpretability.
Lower interpretation cost can make advice easier to use while also reinforcing belief movements that reduce high-value acquisition.

Several questions remain.
One direction is to identify primitive conditions for weak gross improvement in classification problems beyond Shannon acquisition costs.
Another is endogenous AI signal design.
In the baseline model, if the set of designable AI signals is unrestricted, full revelation allows the decision maker to take the statewise optimal action without further acquisition, maximizing both gross and net utility.
The design problem becomes nontrivial when the feasible set of AI signals is restricted.
Empirically, the theory points to measuring how AI advice changes human information acquisition as well as final actions.
Such evidence would help distinguish behavioral misuse from rational reductions or redirections of acquisition induced by AI advice.

\bibliographystyle{splncs04}
\bibliography{ref}

\newpage
\appendix
\renewcommand{\theHsection}{appendix.\arabic{section}}
\renewcommand{\theHsubsection}{appendix.\arabic{section}.\arabic{subsection}}
\renewcommand{\theHsubsubsection}{appendix.\arabic{section}.\arabic{subsection}.\arabic{subsubsection}}

The appendix has two parts.
Appendix~\ref{app:main-proofs} gives analytical proofs.
Appendix~\ref{app:numerical-verifications} records the numerical verifications and illustrations.

\input{appendix_A_proofs_main_text}
\input{appendix_B_numerical_verifications}

\end{document}

%% file: appendix_A_proofs_main_text.tex
\section{Proofs for Main Results}
\label{app:main-proofs}

This appendix proves the results stated in the main text.

\subsection{Concavification of Acquisition}
\label{app:proof-acquisition}

\begin{proof}[of Lemma~\ref{lem:acquisition-cav}]
For any posterior split \(\eta\in\mathcal E_\X(\mu)\), the net objective in Eq.~\eqref{eq:continuation-net-value} can be written as
\[
    \E_{\nu\sim\eta}v(\nu)-\Phi(\mu)+\E_{\nu\sim\eta}\Phi(\nu)
    =
    \E_{\nu\sim\eta}(v+\Phi)(\nu)-\Phi(\mu).
\]
Taking the supremum over posterior splits gives
\[
    W_\Phi(\mu)
    =
    \cav(v+\Phi)(\mu)-\Phi(\mu).
\]
Since the state space is finite and the function \(v+\Phi\) is continuous, Carath\'eodory's theorem allows the supremum to be taken over splits with support size at most \(|\X|+1\).
The feasible set of such weighted supports is compact after allowing zero weights, so the supremum is attained.
The gross-favorable maximum is attained as well.
Indeed, apply Carath\'eodory's theorem to the joint moments \((\nu,(v+\Phi)(\nu),v(\nu))\).
Any feasible net and gross pair can be represented by at most \(|\X|+2\) posterior beliefs.
Compactness of this representation gives attainment of the largest gross payoff among net optimizers.
\end{proof}

\subsection{Posterior-Geometry Criterion}
\label{app:proof-posterior-geometry}

\begin{proof}[of Theorem~\ref{thm:posterior-geometry}]
First suppose the continuation gross value \(G_\Phi\) is convex.
For any AI signal, its posterior split gives
\[
    \sum_{y\in\Y}\alpha_y\mu_y=\mu_0.
\]
Jensen's inequality gives
\[
    G^{\mathrm{AI}}_\Phi(Y,\mu_0)
    =
    \sum_{y\in\Y}\alpha_yG_\Phi(\mu_y)
    \ge
    G_\Phi\left(\sum_{y\in\Y}\alpha_y\mu_y\right)
    =
    G_\Phi(\mu_0).
\]
Thus every AI signal is weakly gross-improving.

Now suppose the continuation gross value is not convex.
Then there exist beliefs \(\mu^1,\mu^2\in\DeltaX\) and a number \(t\in(0,1)\) such that, for \(\mu_0=t\mu^1+(1-t)\mu^2\),
\[
    tG_\Phi(\mu^1)+(1-t)G_\Phi(\mu^2)
    <
    G_\Phi(\mu_0).
\]
The two-point split
\[
    t\delta_{\mu^1}+(1-t)\delta_{\mu^2}
\]
is Bayes plausible at \(\mu_0\), and is implementable by a binary AI signal.
For completeness, one implementing channel is
\[
    P(Y=1\mid X=x)
    =
    \frac{t\mu^1_x}{\mu_{0,x}}
\]
whenever \(\mu_{0,x}>0\), with arbitrary values on states that have zero prior probability.
The inequality \(0\le t\mu^1_x\le\mu_{0,x}\) makes this a valid channel.
The gross utility of this signal is the left side of the strict inequality above, so the signal generates gross harm.
\end{proof}

\subsection{Shannon Advice and Net Utility}
\label{app:proof-shannon-net-benefit}

\begin{corollary}[Shannon Advice Is Net-Beneficial]
\label{cor:shannon-net-benefit}
Suppose the uncertainty function is \(\Phi=\lambda H_\X\) with \(\lambda>0\).
Every AI signal weakly raises optimized net utility.
If such a signal generates gross harm, the paid acquisition cost falls.
\end{corollary}

\begin{proof}
By Lemma~\ref{lem:shannon-dual}, the Shannon continuation net value can be written as
\[
    W_{\lambda H_\X}(\mu)
    =
    \lambda\max_{r\in\Delta(\A)}
    \sum_x\mu(x)\log\sum_a r_a\exp(u(a,x)/\lambda).
\]
For each action prior \(r\), the displayed objective is affine in the belief \(\mu\).
The maximum of affine functions is convex.
Therefore Jensen's inequality gives
\[
    W^{\mathrm{AI}}_{\lambda H_\X}(Y,\mu_0)
    =
    \sum_y\alpha_yW_{\lambda H_\X}(\mu_y)
    \ge
    W_{\lambda H_\X}(\mu_0).
\]
If the signal generates gross harm, then \(\Delta G_\Phi<0\) and \(\Delta W_\Phi\ge0\).
Since the definitions give \(\Delta W_\Phi=\Delta G_\Phi-\Delta K_\Phi\), the acquisition-cost change satisfies \(\Delta K_\Phi<0\).
\end{proof}

\subsection{Binary-State Protection}
\label{app:proof-binary-ups}

\begin{proof}[of Theorem~\ref{thm:binary-ups}]
Identify the binary belief simplex with \([0,1]\) by writing \(p=P(X=1)\), and write \(\mathcal E(p)\) for the set of finite posterior splits over \(q\in[0,1]\) with mean \(p\).
Let
\[
    F(p)=v(p)+\Phi(p),
    \qquad
    \widehat F(p)=\cav F(p).
\]
Since \(C^\Phi\) is posterior-separable,
\[
\begin{aligned}
    W_\Phi(p)
    &=
    \sup_{\eta\in\mathcal E(p)}
    \left\{
        \E_{q\sim\eta}v(q)
        -
        \left(\Phi(p)-\E_{q\sim\eta}\Phi(q)\right)
    \right\}  \\
    &=
    \sup_{\eta\in\mathcal E(p)}
    \E_{q\sim\eta}\bigl[v(q)+\Phi(q)\bigr]
    -
    \Phi(p)
    =
    \widehat F(p)-\Phi(p).
\end{aligned}
\]
Thus the net-optimal acquisition splits at belief \(p\) are the posterior splits that attain the concavification value \(\widehat F(p)\).

Let \([a,b]\) be a maximal non-singleton affine face of \(\widehat F\).
By minimality of the concave majorant, the endpoints are contact beliefs:
\[
    F(a)=\widehat F(a),
    \qquad
    F(b)=\widehat F(b).
\]
For \(p\in[a,b]\), define the endpoint split
\[
    \eta^{a,b}_p
    =
    \frac{b-p}{b-a}\delta_a
    +
    \frac{p-a}{b-a}\delta_b .
\]
It is Bayes plausible at \(p\), and because \(\widehat F\) is affine on \([a,b]\),
\[
\begin{aligned}
    \E_{\eta^{a,b}_p}F(q)
    &=
    \frac{b-p}{b-a}F(a)
    +
    \frac{p-a}{b-a}F(b)  \\
    &=
    \frac{b-p}{b-a}\widehat F(a)
    +
    \frac{p-a}{b-a}\widehat F(b)
    =
    \widehat F(p).
\end{aligned}
\]
Hence \(\eta^{a,b}_p\) is net-optimal.

Now take any net-optimal split \(\eta\) at an interior belief \(p\in(a,b)\).
Then
\[
    \widehat F(p)
    =
    \E_\eta F(q)
    \le
    \E_\eta \widehat F(q)
    \le
    \widehat F\!\left(\E_\eta q\right)
    =
    \widehat F(p).
\]
Both inequalities must bind.
Thus \(\eta\) is supported on contact beliefs of this affine face.
At an endpoint of a maximal affine face, a nondegenerate split cannot attain Jensen equality across a change in slope.
The selected gross value there is \(v\).
Moreover,
\[
    \E_\eta v(q)
    =
    \E_\eta F(q)-\E_\eta\Phi(q)
    =
    \widehat F(p)-\E_\eta\Phi(q).
\]
Therefore, among net-optimal splits, the gross-favorable rule maximizes selected gross utility by minimizing \(\E_\eta\Phi(q)\).

Because \(\Phi\) is concave, for every \(q\in[a,b]\),
\[
    \Phi(q)
    \ge
    \frac{b-q}{b-a}\Phi(a)
    +
    \frac{q-a}{b-a}\Phi(b).
\]
Taking expectations under any split \(\eta\) with mean \(p\) gives
\[
    \E_\eta\Phi(q)
    \ge
    \frac{b-p}{b-a}\Phi(a)
    +
    \frac{p-a}{b-a}\Phi(b),
\]
with equality under the endpoint split \(\eta^{a,b}_p\).
Hence the gross-favorable selected gross value on \([a,b]\) is
\[
    G_\Phi(p)
    =
    \frac{b-p}{b-a}v(a)
    +
    \frac{p-a}{b-a}v(b).
\]
At beliefs that do not belong to any non-singleton affine face of \(\widehat F\), the binding inequalities above force the net-optimal split to be degenerate, so \(G_\Phi(p)=v(p)\).
Thus \(G_\Phi\) is obtained from the static value \(v\) by replacing each maximal affine face \([a,b]\) of \(\widehat F\) with the endpoint chord of \(v\).

It remains to show that this chord replacement preserves convexity.
The static value \(v\) is convex because it is the maximum of finitely many affine functions.
For a replaced interval \([a,b]\), let
\[
    m_{a,b}
    =
    \frac{v(b)-v(a)}{b-a}
\]
be the slope of the replacement chord.
For any \(z<a<b<w\), convexity of \(v\) implies the secant-slope inequalities
\[
    \frac{v(a)-v(z)}{a-z}
    \le
    m_{a,b}
    \le
    \frac{v(w)-v(b)}{w-b}.
\]
Likewise, if \([a,b]\) and \([c,d]\) are two replaced intervals with \(b\le c\), then
\[
    \frac{v(b)-v(a)}{b-a}
    \le
    \frac{v(d)-v(c)}{d-c}.
\]
These are the standard monotonicity inequalities for secant slopes of a convex function.
They imply that the slopes of \(G_\Phi\) are nondecreasing along the belief line: inside unreplaced regions this follows from convexity of \(v\), inside replaced regions the slope is constant, and at the boundaries the displayed inequalities give the correct ordering.
Hence \(G_\Phi\) is convex on \([0,1]\).

By Theorem~\ref{thm:posterior-geometry}, convexity of \(G_\Phi\) implies that all AI signals are weakly gross-improving at all priors.

The same conclusion holds under the minimal-information tie-breaking rule.
Here minimal information means the smallest acquisition cost among net-optimal splits.
On a face \([a,b]\), net-optimality fixes \(W_\Phi(p)\), so minimizing acquisition cost is equivalent to minimizing gross utility.
Let
\[
    \mathcal C_{a,b}
    =
    \{q\in[a,b]:F(q)=\widehat F(q)\}
\]
be the contact set on that face.
The minimal-gross selected value is
\[
    \underline G_\Phi(p)
    =
    \inf_{\eta\in\mathcal E(p),\ \supp(\eta)\subseteq\mathcal C_{a,b}}
    \E_{q\sim\eta}v(q),
    \qquad p\in[a,b].
\]
In one dimension, this is the lower convex envelope of \(v\) restricted to the ordered contact set \(\mathcal C_{a,b}\).
It is formed by secant chords between contact beliefs, and those chord slopes are nondecreasing because \(v\) is convex.
The same slope-pasting argument shows that \(\underline G_\Phi\) is convex.
Thus the binary-state weak gross-improvement conclusion does not rely on the gross-favorable selection convention.
\end{proof}

\subsection{Tools for Shannon Acquisition}
\label{app:shannon-tools}

The following lemmas give the Shannon logit representation and the piecewise-affine convexity criterion used in the classification proofs.

\subsubsection{Shannon Logit Representation}
\label{app:shannon-logit}

\begin{lemma}[Shannon Logit Representation]
\label{lem:shannon-dual}
Under Shannon acquisition cost and $\lambda>0$,
\begin{equation}
W_{\lambda H_\X}(\mu)=\lambda\max_{r\in\Delta(\A)}
\sum_{x\in\X}\mu(x)\log\sum_{a\in\A}r_a\exp(u(a,x)/\lambda).
\label{eq:shannon-dual}
\end{equation}
\end{lemma}

\begin{proof}
It is without loss of generality to use final actions as the acquisition signal.
Replacing any acquisition signal by its induced final action preserves gross utility and weakly lowers Shannon information by data processing.
Thus an experiment can be represented by a conditional action rule \(p(a\mid x)\), with unconditional action prior
\[
    r_a=\sum_x\mu(x)p(a\mid x).
\]
For any such rule and its induced \(r\), the log-sum inequality gives, for each state \(x\),
\[
\sum_a p(a\mid x)u(a,x)
-
\lambda\sum_a p(a\mid x)\log\frac{p(a\mid x)}{r_a}
\le
\lambda\log\sum_a r_a\exp(u(a,x)/\lambda).
\]
After averaging over \(x\), every feasible rule has value at most the right side of Eq.~\eqref{eq:shannon-dual} evaluated at its induced \(r\), and at most the maximum over \(r\in\Delta(\A)\).

Conversely, let \(r^*\) maximize the right side of Eq.~\eqref{eq:shannon-dual}, and define
\[
    Z_x(r^*)=\sum_b r_b^*\exp(u(b,x)/\lambda),
    \qquad
    p^*(a\mid x)=
    \frac{r_a^*\exp(u(a,x)/\lambda)}{Z_x(r^*)}.
\]
The KKT conditions for the maximization over the simplex imply that every active action \(a\) with \(r_a^*>0\) satisfies
\[
    \lambda\sum_x\mu(x)\frac{\exp(u(a,x)/\lambda)}{Z_x(r^*)}=\lambda,
\]
where the multiplier equals \(\lambda\) because the \(r^*\)-weighted average of the active gradients is \(\lambda\).
Therefore
\[
    \sum_x\mu(x)p^*(a\mid x)=r_a^*
\]
for every active action, while both sides are zero for inactive actions.
Thus \(p^*\) is a feasible action rule with unconditional action prior \(r^*\), and it attains equality in the log-sum inequality.
This proves Eq.~\eqref{eq:shannon-dual}.
\end{proof}

\subsubsection{Piecewise-Affine Convexity Criterion}
\label{app:pwa-convexity}

\begin{lemma}[Convexity of Continuous Piecewise-Affine Functions]
\label{lem:pwa-convexity}
Let \(f\) be continuous and affine on each cell of a finite polyhedral subdivision of a convex domain.
Working in the affine hull of the domain, for adjacent full-dimensional cells \(C\) and \(C'\) sharing a facet \(F\), let \(n\) be a normal pointing from \(C\) to \(C'\), and let \(\nabla f_C\) and \(\nabla f_{C'}\) be the affine gradients on the two cells.
Then \(f\) is convex if and only if
\[
    (\nabla f_{C'}-\nabla f_C)\cdot n\ge0
\]
for every adjacent facet.
\end{lemma}

\begin{proof}
If \(f\) is convex, every line crossing a shared facet has nondecreasing one-dimensional slopes, which gives the stated inequality.

Conversely, suppose all adjacent jumps are nonnegative.
First consider a line segment that crosses cell boundaries only through relative interiors of facets.
Its restriction is continuous and piecewise affine, with nondecreasing slopes at every crossing, so it is convex.
Any other segment is a limit of such segments.
Continuity of \(f\) preserves the convexity inequalities in the limit, which proves convexity on the domain.
\end{proof}

\subsection{Classification Slope Test}
\label{app:proof-slope-criterion}

\begin{proof}[of Theorem~\ref{thm:slope-criterion}]
Define
\[
    E_i=\exp(w_i/\lambda),
    \qquad
    L_i=E_i-1,
    \qquad
    g_i=\frac{w_iE_i}{L_i}.
\]
For the classification payoff
\[
    u(a,i)=b_i+w_i\one\{a=i\},
\]
we have
\[
    \sum_{a\in\X}r_a\exp(u(a,i)/\lambda)
    =
    \exp(b_i/\lambda)(1+L_ir_i).
\]
Lemma~\ref{lem:shannon-dual} gives
\begin{equation}
W_{\lambda H_\X}(\mu)
=
\sum_i\mu_i b_i+
\lambda
\max_{r\in\Delta(\X)}
\sum_i\mu_i\log(1+L_ir_i).
\label{eq:weighted-classification-dual-app}
\end{equation}
The baseline term \(\sum_i\mu_i b_i\) is linear in \(\mu\).
It appears in the active-set formula below and cancels from adjacent slope jumps and gross Jensen gaps.

For a nonempty active set \(J\subseteq\X\), define
\[
    N_J=1+\sum_{i\in J}\frac1{L_i},
    \qquad
    \bar g_J=
    \frac1{N_J}
    \sum_{i\in J}\frac{w_iE_i}{L_i^2}.
\]
Any label with zero prior probability is inactive at an optimum, since allocating positive probability to it reduces the objective on the positive-probability labels.
On the remaining coordinates, the objective is strictly concave, so the optimal action prior is unique.
Fix a cell on which the active set is
\[
    J=\{i:r_i^*>0\}.
\]
The active KKT conditions for Eq.~\eqref{eq:weighted-classification-dual-app} are
\[
    \frac{L_i\mu_i}{1+L_ir_i^*}=\tau,
    \qquad i\in J.
\]
Hence
\[
    r_i^*=\frac{\mu_i}{\tau}-\frac1{L_i},
    \qquad
    \tau=\frac{\mu(J)}{N_J},
    \qquad
    \mu(J)=\sum_{i\in J}\mu_i.
\]
In particular,
\[
    1+L_ir_i^*=\frac{L_i\mu_i}{\tau}.
\]
The correct-label probability in state \(i\in J\) is
\[
    \Pr(a=i\mid i)
    =
    \frac{E_ir_i^*}{1+L_ir_i^*}.
\]
Substituting the KKT expression gives
\[
\mu_iw_i\Pr(a=i\mid i)
=
\mu_i\frac{w_iE_i}{L_i}
-
\tau\frac{w_iE_i}{L_i^2}.
\]
Summing over \(i\in J\) and using \(\tau=\mu(J)/N_J\) gives the gross value on active cell \(J\),
\begin{equation}
G_J(\mu)
=
\sum_i\mu_i b_i+
\sum_{i\in J}(g_i-\bar g_J)\mu_i .
\label{eq:gJ-weighted-app}
\end{equation}
An inactive label \(k\notin J\) satisfies
\[
    L_k\mu_k\le \tau=\frac{\mu(J)}{N_J},
\]
and the entry boundary between \(J\) and \(J\cup\{k\}\) is
\begin{equation}
    L_k\mu_k=\frac{\mu(J)}{N_J}.
\label{eq:entry-boundary-app}
\end{equation}

The active and inactive KKT inequalities are affine in \(\mu\), so their closures form a finite polyhedral subdivision of the belief simplex.
The unique optimal action prior varies continuously with the belief, including on faces with zero-probability states.
Thus the gross-value formulas extend continuously to cell boundaries.
Now consider adjacent active sets \(J\) and \(J\cup\{k\}\).
Since
\[
    N_{J\cup\{k\}}=N_J+\frac1{L_k},
\]
Eq.~\eqref{eq:gJ-weighted-app} gives
\begin{equation}
G_{J\cup\{k\}}(\mu)-G_J(\mu)
=
\frac{N_J\Delta_{J,k}}{N_JL_k+1}
\left(
    L_k\mu_k-\frac{\mu(J)}{N_J}
\right).
\label{eq:active-set-difference-slope-app}
\end{equation}
The term in parentheses is an affine coordinate normal to the facet within the belief simplex.
It is zero on the shared facet in Eq.~\eqref{eq:entry-boundary-app} and positive on the \(J\cup\{k\}\) side.
The prefactor outside \(\Delta_{J,k}\) is positive.
Hence the adjacent directional slope jump has the sign of
\[
\Delta_{J,k}
=
g_k-\bar g_J
=
\frac{w_k E_k}{E_k-1}
-
\frac{
\sum_{i\in J}{w_i E_i}/{(E_i-1)^2}
}{
1+\sum_{i\in J}{1}/{(E_i-1)}
}.
\]

If \(\Delta_{J,k}\ge0\) at every feasible adjacent boundary, then Eq.~\eqref{eq:active-set-difference-slope-app} and Lemma~\ref{lem:pwa-convexity} imply that \(G_{\lambda H_\X}\) is convex.
Therefore every Bayes-plausible posterior split has a weakly nonnegative gross Jensen gap.
Theorem~\ref{thm:posterior-geometry} then implies that every AI signal is weakly gross-improving.

If some feasible adjacent boundary has \(\Delta_{J,k}<0\), take a line in the simplex crossing the relative interior of the shared facet in its normal direction.
Let the coordinate \(s<0\) be on the \(J\) side and \(s>0\) be on the \(J\cup\{k\}\) side.
For sufficiently small \(\varepsilon>0\), the equal-weight split between the two posterior beliefs at \(s=-\varepsilon\) and \(s=\varepsilon\) has mean on the boundary.
The gross Jensen gap equals
\[
    \frac{\varepsilon}{2}(m_+-m_-),
\]
where \(m_+-m_-\) is the adjacent slope jump.
Since this jump is negative, the gap is strictly negative.
The linear baseline term cancels in this gap, so \(G_{\lambda H_\X}\) is nonconvex.
Theorem~\ref{thm:posterior-geometry} then implies that the environment admits an AI signal that generates gross harm.
\end{proof}

\subsection{Homogeneous Classification}
\label{app:proof-homogeneous-classification}

\begin{proof}[of Proposition~\ref{prop:homogeneous-classification}]
The payoff is
\[
    u(a,x)=b_x+w\one\{a=x\},
    \qquad w>0.
\]
Set
\[
    E=\exp(w/\lambda),
    \qquad
    L=E-1.
\]
The action-independent term contributes \(\sum_x\mu(x)b_x\).
By Lemma~\ref{lem:shannon-dual},
\[
W_{\lambda H_\X}(\mu)
=
\sum_x\mu(x) b_x+
\lambda
\max_{r\in\Delta(\X)}
\sum_x\mu(x)\log(1+Lr_x).
\]

Fix an active set \(J=\{x:r_x^*>0\}\).
The KKT conditions are
\[
    \frac{L\mu(x)}{1+Lr_x^*}=\tau,
    \qquad x\in J.
\]
Thus
\[
    r_x^*=\frac{\mu(x)}{\tau}-\frac1L,
    \qquad
    \tau=\frac{L\mu(J)}{L+|J|},
    \qquad
    \mu(J)=\sum_{x\in J}\mu(x).
\]
The correct-classification probability in state \(x\in J\) is
\[
    \Pr(a=x\mid x)
    =
    \frac{Er_x^*}{1+Lr_x^*}.
\]
Using \(1+Lr_x^*=L\mu(x)/\tau\), we obtain
\[
\sum_{x\in J}\mu(x) w\Pr(a=x\mid x)
=
\frac{wE\tau}{L}\sum_{x\in J}r_x^*
=
\frac{wE}{L+|J|}\mu(J).
\]
Therefore the gross value on active cell \(J\) is
\[
G_J(\mu)
=
\sum_x\mu(x) b_x+
\frac{wE}{L+|J|}\sum_{x\in J}\mu(x).
\]

It remains to identify which affine expression applies.
The inactive KKT inequalities are
\[
    \mu(x)\le \frac{\mu(J)}{L+|J|},
    \qquad x\notin J,
\]
while active probabilities require
\[
    \mu(x)> \frac{\mu(J)}{L+|J|},
    \qquad x\in J
\]
for positive action probabilities.
These inequalities imply that \(J\) maximizes
\[
    \frac{\mu(T)}{L+|T|}
\]
over nonempty \(T\subseteq\X\).
To see this, write \(c=\mu(J)/(L+|J|)\).
The active coordinates satisfy \(\sum_{x\in J}(\mu(x)-c)=Lc\), and every inactive coordinate has \(\mu(x)-c\le0\).
Consequently, \(\mu(T)-|T|c\le Lc\) for every \(T\), with equality at \(J\).
Boundary labels with zero action probability leave the maximizing value unchanged.
Hence
\[
G_{\lambda H_\X}(\mu)=\sum_x\mu(x) b_x+
\max_{\varnothing\ne J\subseteq\X}
\frac{w E}{L+|J|}\sum_{x\in J}\mu(x).
\]
Substituting \(E=e^{w/\lambda}\) and \(L=e^{w/\lambda}-1\) gives Eq.~\eqref{eq:hom-class-g}.

Finally, weak gross improvement in the homogeneous case follows directly from Theorem~\ref{thm:slope-criterion}.
For any feasible adjacent boundary between \(J\) and \(J\cup\{k\}\), let \(m=|J|\).
The slope jump is
\[
\Delta_{J,k}
=
\frac{wE}{L}
-
\frac{m wE/L^2}{1+m/L}
=
\frac{wE}{L+m}
>0.
\]
Thus every feasible adjacent-boundary slope jump is strictly positive, so every AI signal is weakly gross-improving in homogeneous classification.
\end{proof}

\subsection{Three-State Weighted Classification}
\label{app:proof-three-state-classification}

\begin{proof}[of Proposition~\ref{prop:three-state-gross-improvement}]
The baseline term \(\sum_i\mu_i b_i\) is linear.
By Theorem~\ref{thm:slope-criterion}, it is enough to check entrant boundaries for the stake vector \(w\).
Since \(|\X|=3\), an adjacent entry can only move from an active set of size one to size two, or from an active set of size two to size three.

First let \(J=\{i\}\).
Then
\[
\bar g_J
=
\frac{w_i}{e^{w_i/\lambda}-1}
<
\lambda,
\]
because \(s/(e^s-1)<1\) for every \(s>0\).
For every entrant \(k\),
\[
g_k
=
\lambda
\frac{(w_k/\lambda)e^{w_k/\lambda}}
{e^{w_k/\lambda}-1}
>
\lambda.
\]
Thus \(\Delta_{J,k}=g_k-\bar g_J>0\).

Now let \(J=\{i,j\}\).
The inequality \(\bar g_J<\lambda\) is equivalent to
\[
\sum_{m\in\{i,j\}}
\left[
\frac{(w_m/\lambda)e^{w_m/\lambda}}
{(e^{w_m/\lambda}-1)^2}
-
\frac1{e^{w_m/\lambda}-1}
\right]
<1.
\]
For \(s>0\), define
\[
\psi(s)
=
\frac{s e^s}{(e^s-1)^2}
-
\frac1{e^s-1}
=
\frac{(s-1)e^s+1}{(e^s-1)^2}.
\]
The numerator is positive for \(s>0\).
Moreover, \(\psi(s)<1/2\) is equivalent to
\[
    e^s-e^{-s}>2s,
\]
which holds strictly for \(s>0\).
Hence each term in the preceding sum is strictly below \(1/2\), so \(\bar g_J<\lambda\).
Every entrant still has \(g_k>\lambda\).
Therefore \(\Delta_{J,k}>0\) for every feasible adjacent boundary.

All feasible adjacent-boundary slope jumps are positive.
Theorem~\ref{thm:slope-criterion} implies that every AI signal is weakly gross-improving.
\end{proof}

\subsection{Entry-Dilution Specialization}
\label{app:entry-dilution}

The following corollary specializes the weighted-classification slope criterion to common high-stake blocks.
It explains why the four-class example is the first possible common-high-stake dilution failure.

\begin{corollary}[Common-Stake Entry-Dilution Specialization]
\label{cor:common-stake-entry-dilution-app}
Suppose an active block \(J\) consists of \(m\) states with common normalized stake \(h=w/\lambda\), and an entering state \(k\) has normalized stake \(\ell<h\).
At any feasible adjacent boundary of this form, a downward gross-slope jump occurs if
\begin{equation}
A(\ell)<\frac{mhe^h}{(e^h-1)(e^h-1+m)},
\qquad
A(t)=\frac{te^t}{e^t-1}.
\label{eq:dilution-condition}
\end{equation}
One-high and two-high active blocks cannot satisfy Eq.~\eqref{eq:dilution-condition} for any positive entrant stake, while a three-high active block can.
\end{corollary}

\begin{proof}
Substitute \(w_i/\lambda=h\) for all \(i\in J\) and \(w_k/\lambda=\ell\) into Theorem~\ref{thm:slope-criterion}.
The entrant term is
\[
    g_k=\lambda A(\ell),
\]
and the active-block dilution term is
\[
    \bar g_J
    =
    \lambda
    \frac{mhe^h}{(e^h-1)(e^h-1+m)}.
\]
Thus \(\Delta_{J,k}<0\) is equivalent to Eq.~\eqref{eq:dilution-condition}.

If \(m=1\), the right side is
\[
    \frac{h}{e^h-1}<1.
\]
If \(m=2\), the right side is
\[
    \frac{h}{\sinh h}<1.
\]
Since \(A(\ell)>1\) for every \(\ell>0\), neither case can satisfy Eq.~\eqref{eq:dilution-condition} for any positive entrant stake.
If \(m=3\), the right side can exceed one.
For example, at \(h=1/2\) it is approximately \(1.0448\), while \(A(\ell)\downarrow1\) as \(\ell\downarrow0\).
Hence a sufficiently low-stake entrant can generate a downward jump when three high-stake states are already active.
\end{proof}

\subsection{Nonmonotonicity in Signal Informativeness}
\label{app:proof-informativeness-nonmonotonicity}

\begin{proof}[of Proposition~\ref{prop:informativeness-nonmonotonicity}]
Let \(Y\) be a gross-harmful signal at \(\mu_0\).
Denote the gross values under no AI, under \(Y\), and under full revelation by
\[
    g_0=G_\Phi(\mu_0),
    \qquad
    g_Y=G^{\mathrm{AI}}_\Phi(Y,\mu_0),
    \qquad
    g_F=\sum_x\mu_0(x)v(e_x).
\]
Gross harm gives \(g_Y<g_0\).
Full revelation permits the statewise optimal action, so \(g_Y<g_0\le g_F\).
In both constructions below, the choice of which information to provide is independent of the state, and the message identifies that choice.

First consider a path from no information to \(Y\).
For \(p\in[0,1]\), let \(Y_p^-\) provide \(Y\) with probability \(p\) and otherwise provide a distinct uninformative message.
The path starts from no information at \(p=0\) and reaches \(Y\) at \(p=1\).
For \(p<p'\), replacing each message from the \(Y\) branch of \(Y_{p'}^-\) with the uninformative message with probability \(1-p/p'\) produces \(Y_p^-\).
This garbling establishes Blackwell dominance as \(p\) increases.
The gross value is
\[
    G^{\mathrm{AI}}_\Phi(Y_p^-,\mu_0)
    =(1-p)g_0+pg_Y,
\]
which strictly decreases with \(p\).

Next consider a path from \(Y\) to full revelation.
For \(q\in[0,1]\), let \(Y_q^+\) fully reveal the state with probability \(q\) and otherwise provide \(Y\).
The path starts from \(Y\) at \(q=0\) and reaches full revelation at \(q=1\).
For \(q<q'\), keep every message from the \(Y\) branch of \(Y_{q'}^+\).
For a full-revelation message \(x\), retain that message with probability \(q/q'\).
Otherwise generate a message according to \(P_{Y|X}(\cdot\mid x)\) and report it as a realization of \(Y\).
This garbling produces \(Y_q^+\), so informativeness increases with \(q\).
The gross value is
\[
    G^{\mathrm{AI}}_\Phi(Y_q^+,\mu_0)
    =(1-q)g_Y+qg_F,
\]
which strictly increases with \(q\).

Blackwell-equivalent signals induce the same posterior distribution at \(\mu_0\) and hence the same expected continuation gross value.
The strict value changes above imply strict dominance on each path.
Both signal laws vary continuously with their mixing probabilities, and the paths meet at \(Y\), with the common endpoint included only once.
Transitivity of the Blackwell order also orders signals on opposite sides of this endpoint.
Together the paths form the required continuous family, with gross utility strictly decreasing up to \(Y\) and strictly increasing thereafter.
\end{proof}

\subsection{Gross Harm under Endogenous Interpretation}
\label{app:proof-endogenous-interpretation}

\begin{proof}[of Theorem~\ref{thm:shannon-nonconvexity-harm}]
Since \(G_\Phi\) is not convex, there exist \(\mu^1,\mu^2\in\DeltaX\) and \(t\in(0,1)\) such that, for \(\mu_0=t\mu^1+(1-t)\mu^2\),
\begin{equation}
    G_\Phi(\mu_0)
    >
    tG_\Phi(\mu^1)+(1-t)G_\Phi(\mu^2).
\label{eq:shannon-gross-jensen-loss-app}
\end{equation}
We first show that the same split gives a strict net Jensen gain for the Shannon continuation value.
By Lemma~\ref{lem:shannon-dual},
\[
    W_\Phi(\mu)=\max_{r\in\Delta(\A)}L(r,\mu),
\]
where
\[
    L(r,\mu)=
    \lambda
    \sum_{x\in\X}
    \mu(x)
    \log
    \left(
        \sum_{a\in\A}r_a\exp(u(a,x)/\lambda)
    \right).
\]
For a fixed action prior \(r\), the function \(L(r,\cdot)\) is affine.
Let \(R(\mu)=\argmax_{r\in\Delta(\A)}L(r,\mu)\) be the set of optimal action priors.
For \(r\in\Delta(\A)\), define the associated Shannon action rule and gross-payoff functional by
\[
    \sigma(a\mid x;r)
    =
    \frac{r_a\exp(u(a,x)/\lambda)}{\sum_{b\in\A}r_b\exp(u(b,x)/\lambda)},
    \qquad
    A(r,\mu)=\sum_x\mu(x)\sum_a\sigma(a\mid x;r)u(a,x).
\]
For a fixed action prior \(r\), the function \(A(r,\cdot)\) is affine, and the gross-favorable continuation selection gives
\[
    G_\Phi(\mu)=\max_{r\in R(\mu)}A(r,\mu).
\]
Suppose, toward a contradiction, that
\[
    W_\Phi(\mu_0)
    =
    tW_\Phi(\mu^1)+(1-t)W_\Phi(\mu^2).
\]
Choose \(r^*\in R(\mu_0)\) attaining \(G_\Phi(\mu_0)=A(r^*,\mu_0)\).
Since \(L(r^*,\cdot)\) is affine and \(L(r^*,\mu^j)\le W_\Phi(\mu^j)\) for \(j=1,2\), the displayed equality forces \(r^*\in R(\mu^1)\cap R(\mu^2)\).
Therefore
\[
\begin{aligned}
    G_\Phi(\mu_0)
    &= A(r^*,\mu_0) \\
    &= tA(r^*,\mu^1)+(1-t)A(r^*,\mu^2) \\
    &\le tG_\Phi(\mu^1)+(1-t)G_\Phi(\mu^2),
\end{aligned}
\]
contradicting Eq.~\eqref{eq:shannon-gross-jensen-loss-app}.
Hence
\begin{equation}
    W_\Phi(\mu_0)
    <
    tW_\Phi(\mu^1)+(1-t)W_\Phi(\mu^2).
\label{eq:shannon-net-jensen-gain-app}
\end{equation}

Construct a binary AI signal \(Y\in\{1,2\}\) with \(P(Y=1)=t\), \(P(Y=2)=1-t\), \(P(X\mid Y=1)=\mu^1\), and \(P(X\mid Y=2)=\mu^2\).
Then \(\bar q=(t,1-t)\).
For an interpretation posterior \(q=(s,1-s)\), the induced payoff-state belief is
\[
    \rho(q)=s\mu^1+(1-s)\mu^2.
\]
Set
\[
    F(s)=W_\Phi\bigl(s\mu^1+(1-s)\mu^2\bigr),
    \qquad s\in[0,1].
\]
The function \(F\) is convex because \(W_\Phi\) is the maximum of affine functions.
Eq.~\eqref{eq:shannon-net-jensen-gain-app} gives
\[
    F(t)<tF(1)+(1-t)F(0).
\]
Let \(\ell(s)=sF(1)+(1-s)F(0)\) be the endpoint chord.
Convexity gives \(F(s)\le\ell(s)\) for all \(s\in[0,1]\), and the strict inequality at \(t\) implies \(F(s)<\ell(s)\) for all \(s\in(0,1)\).
Thus, among interpretation posterior splits with mean \(t\), full interpretation
\[
    \zeta^F=t\delta_1+(1-t)\delta_0
\]
is the unique zero-cost net optimizer.
Its selected gross utility is
\[
    tG_\Phi(\mu^1)+(1-t)G_\Phi(\mu^2),
\]
which is strictly below \(G_\Phi(\mu_0)\) by Eq.~\eqref{eq:shannon-gross-jensen-loss-app}.
This proves Theorem~\ref{thm:shannon-nonconvexity-harm}.

The same argument also establishes robustness to a small positive interpretation cost, a result not used in the main text.
Choose any nonaffine continuous concave \(\Psi:\DeltaY\to\R\).
Let \(\mathcal P_t\) be the compact set of probability measures on \([0,1]\) with mean \(t\).
For \(\zeta\in\mathcal P_t\), write
\[
    B(\zeta)=\int F(s)\,d\zeta(s),
    \qquad
    I_\Psi(\zeta)=\Psi(\bar q)-\int\Psi(s,1-s)\,d\zeta(s),
\]
and
\[
    \mathcal G(\zeta)=
    \int G_\Phi\bigl(s\mu^1+(1-s)\mu^2\bigr)\,d\zeta(s).
\]
The first-stage net objective under interpretation cost \(\varepsilon\Psi\) is \(B(\zeta)-\varepsilon I_\Psi(\zeta)\).
The function \(L\) is continuous on the compact action-prior simplex, so \(R(\mu)\) is compact-valued and upper hemicontinuous.
Continuity of \(A\) then makes \(G_\Phi(\mu)=\max_{r\in R(\mu)}A(r,\mu)\) upper semicontinuous.
Hence \(\mathcal G\) is upper semicontinuous.
Since \(\zeta^F\) is the unique maximizer of \(B\) and \(\mathcal G(\zeta^F)<G_\Phi(\mu_0)\), there is a neighborhood \(N\) of \(\zeta^F\) and a number \(\delta>0\) such that
\[
    \mathcal G(\zeta)
    \le
    G_\Phi(\mu_0)-\delta
    \qquad
    \text{for all }\zeta\in N.
\]
If \(\mathcal P_t\setminus N\) is nonempty, compactness and uniqueness give
\[
    \gamma
    =
    B(\zeta^F)-\max_{\zeta\in\mathcal P_t\setminus N}B(\zeta)
    >0.
\]
The cost difference \(|I_\Psi(\zeta^F)-I_\Psi(\zeta)|\) is bounded on \(\mathcal P_t\).
Choose \(\bar\varepsilon>0\) so small that the perturbation by \(\varepsilon I_\Psi\) is less than \(\gamma\) on \(\mathcal P_t\setminus N\) for every \(0\le\varepsilon<\bar\varepsilon\).
If \(\mathcal P_t\setminus N\) is empty, take any \(\bar\varepsilon>0\).
Then every first-stage optimizer lies in \(N\) for every \(0\le\varepsilon<\bar\varepsilon\).
Consequently every gross-favorable extension optimizer has gross utility below \(G_\Phi(\mu_0)\) for all sufficiently small positive interpretation-cost multipliers.
\end{proof}

\subsection{Cheaper Interpretation Can Lower Gross Utility}
\label{app:proof-cheaper-interpretation}

\begin{proof}[of Proposition~\ref{prop:cheaper-interpretation-harms}]
Choose a Shannon acquisition environment with nonconvex \(G_\Phi\), such as Example~\ref{ex:treatment}.
Theorem~\ref{thm:shannon-nonconvexity-harm} supplies a prior and a binary AI signal \(Y\) for which full interpretation is the unique zero-cost optimal posterior split and
\[
    G^{\mathrm{AI}}_{0,\Phi}(Y,\mu_0)<G_\Phi(\mu_0).
\]
Set \(\Psi=H_\Y\) and fix any finite \(\lambda_1>\lambda\).

Consider an interpretation channel \(Y\to Z\) followed by a continuation experiment \(S\) chosen conditional on \(Z\).
The joint distribution of \((X,Z,S)\) can be implemented by a single experiment about \(X\), without interpreting \(Y\), through the channel
\[
    P_{Z,S\mid X}(z,s\mid x)
    =\sum_yP_{Y\mid X}(y\mid x)P_{Z\mid Y}(z\mid y)
             P_{S\mid X,Z}(s\mid x,z).
\]
This implementation preserves the final posterior distribution and the gross payoff.
The two-stage policy costs
\[
    \lambda_1 I(Y;Z)+\lambda I(X;S\mid Z),
\]
whereas direct acquisition costs
\[
    \lambda I(X;Z,S)
    =\lambda I(X;Z)+\lambda I(X;S\mid Z).
\]
The Markov chain \(X-Y-Z\) gives \(I(X;Z)\le I(Y;Z)\) by the data processing inequality.
The two-stage cost minus the direct cost is at least
\[
    (\lambda_1-\lambda)I(Y;Z),
\]
which is positive whenever \(I(Y;Z)>0\).
No such informative interpretation can be optimal.
If \(I(Y;Z)=0\), \(Z\) is independent of \(Y\) and of \(X\), so every continuation problem starts from \(\mu_0\).
Optimizing the continuation policy and applying the gross-favorable selection rule gives \(G^{\mathrm{AI}}_{\lambda_1\Psi,\Phi}(Y,\mu_0)=G_\Phi(\mu_0)\).
Combining this equality with the strict loss at zero interpretation cost proves the claim.
\end{proof}

%% file: appendix_B_numerical_verifications.tex
\section{Numerical Verifications}
\label{app:numerical-verifications}

This appendix records numerical verifications for the examples and the illustrations in Section~\ref{sec:discussion-extension}.
Each verification states the environment, posterior split or path, value comparisons, and optimality checks.

\subsection{Shannon Logit KKT Conditions}
\label{app:shannon-kkt}

\begin{lemma}[Shannon Logit KKT Conditions]
\label{lem:kkt-app}
Let the uncertainty function be \(\Phi=\lambda H_\X\) with \(\lambda>0\).
Set \(E_{xa}=e^{u(a,x)/\lambda}\) and \(N_x(r)=\sum_a r_aE_{xa}\).
For an action prior \(r\in\Delta(\A)\), define
\[
    \mathfrak g_a(r,\mu)
    =
    \lambda\sum_x\mu_x\frac{E_{xa}}{N_x(r)}.
\]
If \(\mathfrak g_a(r,\mu)=\lambda\) for every active action \(a\) with \(r_a>0\), and \(\mathfrak g_a(r,\mu)\le\lambda\) for every inactive action, then \(r\) is a global maximizer of the Shannon logit dual.
If all inactive inequalities are strict and the active-support Hessian is nonsingular on the support simplex, the active support is locally stable.
\end{lemma}

\begin{proof}
The Shannon logit dual objective is
\[
    \lambda\sum_x\mu_x\log N_x(r).
\]
It is concave in the action prior \(r\).
Its partial derivative with respect to the action probability \(r_a\) is \(\mathfrak g_a(r,\mu)\).
Under the simplex constraint, the KKT conditions require active gradients to equal a common multiplier and inactive gradients to be weakly below it.
Since
\[
    \sum_a r_a\mathfrak g_a(r,\mu)=\lambda,
\]
the multiplier is \(\lambda\).
Concavity makes the conditions sufficient for global optimality.
Strict inactive slack and nonsingularity give local support stability by the implicit function theorem.
The same argument permits joint feasible perturbations of the belief, payoffs, and positive multiplier \(\lambda\).
Active probabilities and strict inactive inequalities persist, and the selected gross value is locally continuous.
A strict gross-loss margin then persists for nearby Bayes-plausible AI splits.
\end{proof}

Given an action prior \(r\) satisfying these conditions, the conditional action rule is
\[
    \sigma(a\mid x)
    =
    \frac{r_ae^{u(a,x)/\lambda}}
    {\sum_b r_be^{u(b,x)/\lambda}}.
\]
The associated values are
\begin{align*}
    W_{\lambda H_\X}(\mu)
    &=
    \lambda\sum_x\mu_x\log N_x(r),\\
    G_{\lambda H_\X}(\mu)
    &=
    \sum_x\mu_x\sum_a\sigma(a\mid x)u(a,x),\\
    K_{\lambda H_\X}(\mu)
    &=
    G_{\lambda H_\X}(\mu)-W_{\lambda H_\X}(\mu).
\end{align*}

\FloatBarrier
\subsection{Three-State Treatment Verification}
\label{app:treatment-verification}

\begin{figure}[t]
\centering
\includegraphics[width=.96\linewidth]{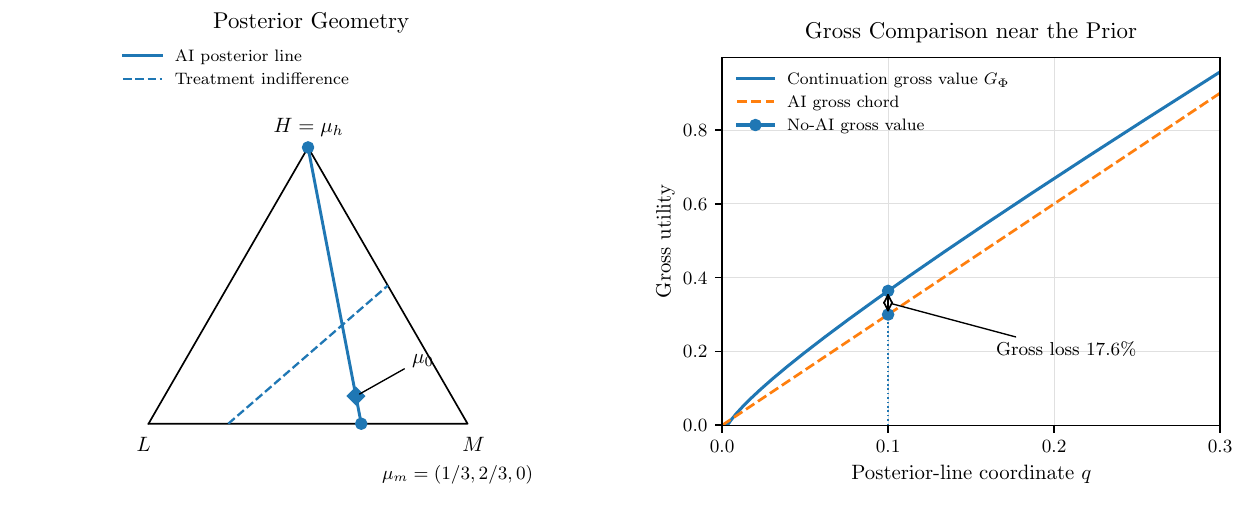}
\caption{Posterior geometry in Example~\ref{ex:treatment} with \(\Phi=H_\X\) and prior \(\mu_0=(0.3,0.6,0.1)\).
The left panel shows the two AI posteriors, their prior mean, and the static treatment-indifference line \(\mu_L-3\mu_M+3\mu_H=0\).
The right panel compares the continuation gross value with the AI gross chord along \(\mu(q)=(1-q)\mu_m+q\mu_h\), focusing on the region near \(q=0.1\).
At the prior, the chord lies below continuation gross value by approximately \(0.06409\), a gross loss of \(17.6\%\).}
\label{fig:posterior-geometry-treatment}
\end{figure}

\begin{proof}[of the claims in Example~\ref{ex:treatment}]
The uncertainty function is \(\Phi=H_\X\), so the Shannon multiplier is \(\lambda=1\).
The prior is \(\mu_0=(0.3,0.6,0.1)\), and treatment payoffs relative to no treatment are \((\Delta_L,\Delta_M,\Delta_H)=(1,-3,3)\).
Write \(r\) for the unconditional treatment probability and \(D_x(r)=1+r(e^{\Delta_x}-1)\).
By Lemma~\ref{lem:shannon-dual}, the no-AI problem is
\[
    W_\Phi(\mu_0)=\max_{r\in[0,1]} f(r),
    \qquad
    f(r)=\sum_{x\in\X}\mu_0(x)\log D_x(r).
\]
The derivatives satisfy
\[
    f'(r)=\sum_x\mu_0(x)\frac{e^{\Delta_x}-1}{D_x(r)},
    \qquad
    f''(r)=-\sum_x\mu_0(x)
    \frac{(e^{\Delta_x}-1)^2}{D_x(r)^2}<0.
\]
The derivative is positive at \(0.22769923828\) and negative at \(0.22769923829\).
Strict concavity gives a unique optimum between these two values, with \(r^*\approx0.227699238283\).
The induced action rule is
\[
    P(1\mid x)=\frac{r^*e^{\Delta_x}}{D_x(r^*)},
    \qquad
    \bigl(P(1\mid L),P(1\mid M),P(1\mid H)\bigr)
    \approx(0.44488773,0.01446649,0.85553028).
\]
Substitution gives the no-AI values in the table below.
Interval evaluation over the root bracket also gives \(0.36408572<G_\Phi(\mu_0)<0.36408574\), so the gross-loss sign does not depend on the displayed rounding.

The AI signal identifies the severe state.
Its posteriors are \(\mu_h=(0,0,1)\) and \(\mu_m=(1/3,2/3,0)\), with probabilities \(0.1\) and \(0.9\).
At \(\mu_h\), treatment gives payoff \(3\) without acquisition.
At \(\mu_m\), the Shannon KKT condition for no acquisition with action \(0\) is strict,
\[
    \sum_x\mu_m(x)e^{\Delta_x}
    =\frac13e+\frac23e^{-3}
    \approx0.93928532<1.
\]
Thus the unique optimal action prior at \(\mu_m\) assigns probability one to no treatment, with net value, gross value, and acquisition cost all equal to zero.
Averaging these two branches gives the following values.
The displayed continuation values are rounded to eight decimal places.
The AI averages are exact.
\[
\begin{array}{lrrr}
\toprule
\text{Starting belief or comparison} & W_\Phi & G_\Phi & K_\Phi\\
\midrule
\text{No AI }(\mu_0) & 0.12040607 & 0.36408573 & 0.24367965\\
\text{Severe posterior }(\mu_h) & 3 & 3 & 0\\
\text{Nonsevere posterior }(\mu_m) & 0 & 0 & 0\\
\text{AI average} & 0.30000000 & 0.30000000 & 0\\
\text{AI minus no AI} & 0.17959393 & -0.06408573 & -0.24367965\\
\bottomrule
\end{array}
\]
The values satisfy \(\Delta W_\Phi=\Delta G_\Phi-\Delta K_\Phi\).
The relative gross loss is \(-\Delta G_\Phi/G_\Phi(\mu_0)\approx0.17601823\), or \(17.6\%\).
The no-AI optimum is interior with negative second derivative.
Both AI posteriors have strict inactive-action slack, equal to \(1-e^{-3}\) at \(\mu_h\) and \(1-(e+2e^{-3})/3\) at \(\mu_m\).
These conditions support the local-robustness statement by Lemma~\ref{lem:kkt-app}.
\end{proof}

\FloatBarrier
\subsection{Four-Class Triage Dilution Verification}
\label{app:triage-verification}

\begin{figure}[t]
\centering
\includegraphics[width=.96\linewidth]{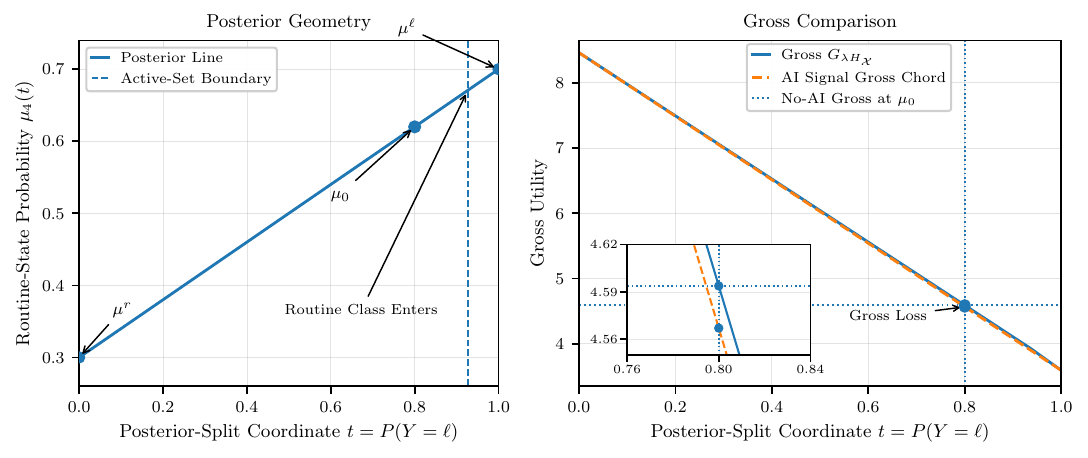}
\caption{Posterior geometry in Example~\ref{ex:weighted-four-state}.
The left panel reduces the symmetric four-state belief line to the routine-state probability \(\mu_4\) and marks the AI posteriors, the prior, and the active-set boundary at which the low-stake routine class enters the action support.
The right panel plots the continuation gross value \(G_{\lambda H_\X}\), with \(\lambda=100\), the AI advice gross chord, and the no-AI gross utility at the prior.
At the prior, the gap between \(G_{\lambda H_\X}(\mu_0)\) and the AI advice chord is the gross loss.}
\label{fig:weighted-triage-geometry-app}
\end{figure}

\begin{proof}[of the claims in Example~\ref{ex:weighted-four-state}]
We verify the four-class triage dilution example.
The payoff is
\[
    u(a,x)=w_x\one\{a=x\},
    \qquad
    w=(30,30,30,5),
\]
and the uncertainty function is \(\Phi=\lambda H_\X\) with \(\lambda=100\).
The prior and AI posteriors are
\[
    \mu_0=
    \left(\frac{19}{150},\frac{19}{150},\frac{19}{150},\frac{31}{50}\right),
\]
\[
    \mu^\ell=
    \left(\frac{1}{10},\frac{1}{10},\frac{1}{10},\frac{7}{10}\right),
    \qquad
    \mu^r=
    \left(\frac{7}{30},\frac{7}{30},\frac{7}{30},\frac{3}{10}\right),
\]
with \(P(Y=\ell)=4/5\) and \(P(Y=r)=1/5\).
Hence
\[
    \mu_0=\frac45\mu^\ell+\frac15\mu^r .
\]

For a candidate active set \(J\subseteq\A\), the Shannon KKT system is
\[
    \pi(a\mid x)
    =
    \frac{r_a\exp(u(a,x)/\lambda)}
    {\sum_{b\in J}r_b\exp(u(b,x)/\lambda)},
    \qquad
    r_a=\sum_{x\in\X}\mu(x) \pi(a\mid x),
    \qquad a\in J.
\]
For each inactive action \(a\notin J\), the corresponding KKT inequality is
\[
    \sum_{x\in\X}
    \mu(x)
    \frac{\exp(u(a,x)/\lambda)}
    {\sum_{b\in J}r_b\exp(u(b,x)/\lambda)}
    \leq 1.
\]

At the no-AI prior, the active set is \(J_0=\{1,2,3\}\), with
\[
    r(\mu_0)=(1/3,1/3,1/3,0).
\]
The inactive-action KKT quantity for action \(4\) is
\[
    0.99210093<1,
\]
so action \(4\) is inactive.
The resulting continuation values are
\[
    W_{\lambda H_\X}(\mu_0)=4.19162455,
    \qquad
    G_{\lambda H_\X}(\mu_0)=4.59374299,
    \qquad
    K_{\lambda H_\X}(\mu_0)=0.40211843.
\]

On the routine-heavy branch \(\mu^\ell\), the active set is \(J_\ell=\{1,2,3,4\}\), with action-prior vector
\[
    r(\mu^\ell)
    =
    (0.04960951,\ 0.04960951,\ 0.04960951,\ 0.85117147).
\]
The resulting branchwise continuation values are
\[
    W_{\lambda H_\X}(\mu^\ell)=3.50627804,
    \qquad
    G_{\lambda H_\X}(\mu^\ell)=3.59329234,
    \qquad
    K_{\lambda H_\X}(\mu^\ell)=0.08701431.
\]

On the high-stake branch \(\mu^r\), the active set is \(J_r=\{1,2,3\}\), with
\[
    r(\mu^r)=(1/3,1/3,1/3,0).
\]
The inactive-action KKT quantity for action \(4\) is
\[
    0.94227342<1,
\]
so action \(4\) remains inactive.
The resulting branchwise continuation values are
\[
    W_{\lambda H_\X}(\mu^r)=7.72141365,
    \qquad
    G_{\lambda H_\X}(\mu^r)=8.46215813,
    \qquad
    K_{\lambda H_\X}(\mu^r)=0.74074448.
\]

After observing the AI signal, the continuation values are
\[
    W^{\mathrm{AI}}_{\lambda H_\X}(Y,\mu_0)
    =
    \frac45 W_{\lambda H_\X}(\mu^\ell)+\frac15 W_{\lambda H_\X}(\mu^r)
    =
    4.34930516,
\]
\[
    G^{\mathrm{AI}}_{\lambda H_\X}(Y,\mu_0)
    =
    \frac45 G_{\lambda H_\X}(\mu^\ell)+\frac15 G_{\lambda H_\X}(\mu^r)
    =
    4.56706550,
\]
and the paid acquisition cost is
\[
    \frac45 K_{\lambda H_\X}(\mu^\ell)+\frac15 K_{\lambda H_\X}(\mu^r)
    =
    0.21776034.
\]
For the AI signal, the differences relative to the no-AI benchmark are
\[
    \Delta W_\Phi
    =
    W^{\mathrm{AI}}_\Phi(Y,\mu_0)-W_\Phi(\mu_0)
    =
    0.15768061>0,
\]
\[
    \Delta G_\Phi
    =
    G^{\mathrm{AI}}_\Phi(Y,\mu_0)-G_\Phi(\mu_0)
    =
    -0.02667749<0,
\]
and
\[
    \Delta K_\Phi
    =
    \frac45 K_\Phi(\mu^\ell)+\frac15 K_\Phi(\mu^r)-K_\Phi(\mu_0)
    =
    -0.18435809<0.
\]
Thus the AI signal strictly raises net value and generates gross harm.
The mechanism is dilution.
The frequent routine-heavy branch activates the low-stake class and shifts the action prior toward it, which saves acquisition cost while lowering gross classification payoff.
The displayed decimals are rounded.
The active probabilities are positive, the inactive inequalities are strict, and the dual Hessian is negative definite on each active simplex.
Together with the strict gross loss, these conditions imply local robustness by Lemma~\ref{lem:kkt-app}.
\end{proof}

\FloatBarrier
\subsection{Safe-Option and Informativeness Verification}
\label{app:safe-option-informativeness}

The safe-option environment used in Fig.~\ref{fig:informativeness-gross} has three states and three actions.
Rows index states and columns index actions.
The payoff matrix is
\[
    U=
    \begin{pmatrix}
    6&3&-12\\
    1&3&1\\
    -12&3&6
    \end{pmatrix},
    \qquad
    \lambda=4.
\]
The middle action is a safe default.
It avoids catastrophic mismatch but does not match either extreme state well.
Input posterior coordinates are given to six decimal places.
Continuation values below are rounded for display.
The prior is
\[
    \mu_0=(0.204151,0.091782,0.704067).
\]
Let \(Y'\) denote the null AI signal, which has a single message and posterior belief \(\mu_0\).
Let \(Y''\) denote the coarse binary AI signal with equal message probabilities and posterior beliefs
\[
    \mu^L=(0.386988,0.148216,0.464796),
    \qquad
    \mu^R=(0.021314,0.035348,0.943338).
\]
Their average is \(\mu_0\).

Support enumeration using Lemma~\ref{lem:kkt-app} gives
\[
\begin{aligned}
    r(\mu^L)&\approx(0,0.91317233,0.08682767),\\
    r(\mu^R)&\approx(0.01135764,0,0.98864236),\\
    r(\mu_0)&\approx(0.21859845,0,0.78140155).
\end{aligned}
\]
The maximum active-gradient residual across these three rows is below \(2\times10^{-14}\).
The minimum inactive slack is approximately \(3.22036\times10^{-4}\), and the active-support negative Hessians are positive definite.
The supplied script checks these global KKT conditions and reports the following continuation values:
\[
\begin{array}{c|rrr}
 &W&G&K\\
\hline
\mu^L&3.01434625&3.20453274&0.19018649\\
\mu^R&5.45655154&5.63248563&0.17593409\\
\mu_0&3.64528931&5.36147630&1.71618698
\end{array}
\]
For the coarse binary signal \(Y''\), endpoint averages give
\[
    \Delta W_\Phi>0.5901,
    \qquad
    \Delta G_\Phi<-0.9429,
    \qquad
    \Delta K_\Phi<-1.5331.
\]
Thus the coarse binary signal strictly raises net value and generates gross harm.

The analytical decrease-then-increase construction is proved in Appendix~\ref{app:proof-informativeness-nonmonotonicity}.
Here the coarse-signal gross value is approximately \(4.4185\), compared with the no-AI value \(5.3615\), a gross loss of about \(17.6\%\).

For visualization, Fig.~\ref{fig:informativeness-gross} plots a two-segment Blackwell-increasing path through the same environment.
The horizontal coordinate is a Blackwell path parameter \(\tau\), not normalized mutual information.
On the first segment, \(\tau\in[0,1]\), the posterior split is
\[
    (1-\tau)\delta_{\mu_0}
    +
    \frac{\tau}{2}\delta_{\mu^L}
    +
    \frac{\tau}{2}\delta_{\mu^R}.
\]
Thus \(\tau=0\) is \(Y'\) and \(\tau=1\) is \(Y''\).
On the second segment, \(\tau=1+\beta\) with \(\beta\in[0,1]\), the posterior split is
\[
    \frac{1-\beta}{2}\delta_{\mu^L}
    +
    \frac{1-\beta}{2}\delta_{\mu^R}
    +
    \beta\sum_{x\in\X}\mu_{0,x}\delta_{e_x}.
\]
The less informative experiment is recovered from the more informative one by the corresponding garbling along each segment.
At the coarse binary signal,
\[
    I(X;Y'')/H(X)\approx0.2021,
\]
reported only as a separate statistic.
At full revelation, the gross utility is approximately
\[
    \sum_x\mu_{0,x}v(e_x)\approx5.7247.
\]

\subsection{Distance-Payoff Interpretation Illustration}
\label{app:distance-interpretation}

The analytical results on endogenous interpretation are proved in Appendices~\ref{app:proof-endogenous-interpretation} and~\ref{app:proof-cheaper-interpretation}.

The numerical illustration in Fig.~\ref{fig:costly-interpretation-gross} uses the distance-payoff environment
\[
    U=
    \begin{pmatrix}
    4&3&0\\
    3&4&3\\
    0&3&4
    \end{pmatrix},
    \qquad
    \Phi=1.5H_\X.
\]
The AI signal is binary with equal probabilities and posterior beliefs
\[
    \mu^L=(0.538308,0.074776,0.386916),
    \qquad
    \mu^R=(0.087382,0.000055,0.912563).
\]
The prior is their average,
\[
    \mu_0=(0.312845,0.0374155,0.6497395).
\]
For an interpretation posterior \(q=P(Y=L\mid Z)\), the continuation belief is
\[
    \mu(q)=q\mu^L+(1-q)\mu^R.
\]
With interpretation uncertainty function \(\Psi=H_\Y\), the interpretation objective at multiplier \(\lambda_1\) is the concavification at \(q=1/2\) of
\[
    W_{1.5H_\X}(\mu(q))+\lambda_1H_\Y(q),
\]
minus \(\lambda_1\log2\).
This expression is Eq.~\eqref{eq:extension-value} specialized to a binary AI signal.

The stage-two Shannon continuation problem is solved on a \(5001\)-point \(q\)-grid using the logit dual and Lemma~\ref{lem:kkt-app}.
The maximum active-gradient residual on this grid is below \(1.4\times10^{-14}\), and all inactive-gradient inequalities hold with nonnegative slack.
The smallest inactive slack on the grid is approximately \(6.03\times10^{-6}\).
The first-stage interpretation problem is solved on \(551\) values of \(\lambda_1\in[0,0.55]\) with step size \(0.001\), using the upper concave hull of \(W_{1.5H_\X}(\mu(q))+\lambda_1H_\Y(q)\).
The supplied script recomputes this upper hull and the rows below.
The following representative rows record the selected interpretation chord, interpretation cost, gross utility, and net gain:
\[
\begin{array}{c|cccccc}
\lambda_1 & q^- & q^+ & \theta^- & C^{\lambda_1H_\Y} & G^{\mathrm{AI}} & \text{net gain}\\
\hline
0.000 & 0.0000 & 1.0000 & 0.5000 & 0.0000 & 3.4674 & 0.2257\\
0.100 & 0.0000 & 0.9998 & 0.4999 & 0.0692 & 3.4673 & 0.1564\\
0.180 & 0.0024 & 0.9914 & 0.4969 & 0.1188 & 3.4622 & 0.1019\\
0.227 & 0.0092 & 0.9744 & 0.4915 & 0.1378 & 3.4521 & 0.0720\\
0.300 & 0.0380 & 0.9170 & 0.4744 & 0.1399 & 3.4788 & 0.0324\\
0.400 & 0.2146 & 0.7050 & 0.4180 & 0.0491 & 3.5991 & 0.0010\\
0.414 & 0.5000 & 0.5000 & 1.0000 & 0.0000 & 3.7124 & 0.0000
\end{array}
\]

The no-AI gross value is
\[
    G_{1.5H_\X}(\mu_0)\approx3.7124.
\]
At \(\lambda_1=0\), full interpretation is selected on the reported grid, and the extension gross utility is
\[
    G^{\mathrm{AI}}_{0\cdot H_\Y,1.5H_\X}(Y,\mu_0)\approx3.4674.
\]
This value is strictly below the no-AI gross value, so endogenous interpretation can generate gross harm in the numerical instance.

The lowest reported grid value is approximately \(3.4521\) at \(\lambda_1=0.227\).
At \(\lambda_1=0.414\), the grid calculation selects the no-interpretation policy and returns the no-AI gross value.
These grid observations illustrate the intermediate-cost pattern and do not establish a global cutoff or prove Proposition~\ref{prop:cheaper-interpretation-harms}.
The analytical proof in Appendix~\ref{app:proof-cheaper-interpretation} establishes no informative interpretation for every \(\lambda_1>1.5\) in this Shannon cost pair.
The small-positive-cost argument following Theorem~\ref{thm:shannon-nonconvexity-harm} explains why endogenous gross harm need not be confined to zero interpretation cost.

\subsection{Quadratic/Gini Treatment with Gross and Net Harm}
\label{app:gini-treatment}

Consider the treatment payoff in Example~\ref{ex:treatment}, with \(\X=\{L,M,H\}\), \(\A=\{0,1\}\), no-treatment payoff \(u(0,x)=0\), and treatment payoffs \((u(1,L),u(1,M),u(1,H))=(1,-3,3)\).
The payoffs are unchanged, but the prior, AI signal, and uncertainty function differ.
The Quadratic/Gini uncertainty function and its UPS cost are
\[
    \Phi_\gamma^G(\mu)=\gamma\left(1-\sum_x\mu_x^2\right),
    \qquad
    C^{\Phi_\gamma^G}(\mu;\eta)
    =\gamma\left(\E_{p\sim\eta}\sum_x p_x^2-\sum_x\mu_x^2\right),
\]
with \(\gamma=2\).
The prior and equally likely AI posteriors are
\[
    \mu_0=\left(\frac{13}{40},\frac38,\frac3{10}\right),
    \qquad
    \mu^h=\left(0,\frac25,\frac35\right),
    \qquad
    \mu^\ell=\left(\frac{13}{20},\frac7{20},0\right).
\]
Their average is \(\mu_0\).
An equivalent signal structure is \(P(h\mid L)=0\), \(P(h\mid M)=8/15\), and \(P(h\mid H)=1\).

\paragraph{Optimal acquisition and verification.}
For each starting belief, the optimal experiment uses two posteriors.
Let \(p^0\) be the posterior followed by no treatment, \(p^1\) the posterior followed by treatment, and \(t=P(a=1)\).
The following table reports the optimal splits.
Decimals are rounded, and fractions are exact.
\[
\begin{array}{ccll}
\toprule
\text{Belief} & t & p^0 & p^1\\
\midrule
\mu_0 & 0.512748885 &
(0.230376313,0.769623687,0) &
(0.414918279,0,0.585081721)\\[1mm]
\mu^h & 19/30 & (0,7/8,1/8) & (0,1/8,7/8)\\[1mm]
\mu^\ell & 3/10 & (1/2,1/2,0) & (1,0,0)\\
\bottomrule
\end{array}
\]
To specify the first row without rounding, let \(t_0\in(0,1)\) solve
\[
    \frac{9}{16(1-t_0)^2}-\frac{9}{25t_0^2}=1.
\]
The left side is strictly increasing from \(-\infty\) to \(+\infty\), so the solution is unique.
It lies between \(0.51274888481\) and \(0.51274888482\).
Set \(a=1-3/[8(1-t_0)]\) and \(b=1-3/(10t_0)\).
The exact first-row posteriors are \(p^0=(a,1-a,0)\) and \(p^1=(b,0,1-b)\), with treatment probability \(t_0\).
Every row satisfies \((1-t)p^0+tp^1=\mu\).

Global optimality follows from a common affine majorant of
\[
    f_0(p)=\Phi_2^G(p),
    \qquad
    f_1(p)=u(1,\cdot)\cdot p+\Phi_2^G(p).
\]
For each starting belief, take \(L_\mu(p)=\ell_\mu\cdot p+c_\mu\).
At the prior, exact coefficients are
\[
    \ell_{\mu_0}=(2-8b,\,8a-8b-2,\,0),
    \qquad
    c_{\mu_0}=4a^2-8a+8b+2.
\]
For the two AI posteriors, take
\[
    (\ell_{\mu^h},c_{\mu^h})=((2,-3,0),49/16),
    \qquad
    (\ell_{\mu^\ell},c_{\mu^\ell})=((0,0,7),1).
\]
These coefficients satisfy \(f_j(p^j)=L_\mu(p^j)\) for \(j=0,1\) and the simplex KKT conditions
\[
    \nabla_i f_j(p^j)-\ell_{\mu,i}
    \begin{cases}
        =\kappa_j, & p_i^j>0,\\
        \le\kappa_j, & p_i^j=0.
    \end{cases}
\]
Both functions \(f_j\) are strictly concave.
The KKT conditions therefore make \(p^j\) the unique global maximizer of \(f_j-L_\mu\), with maximum zero.
Thus the affine function majorizes both pieces everywhere and touches at the stated posteriors.
Bayes plausibility then establishes global optimality of each split.
The two contact points and their mixing weight are unique, so the reported values do not depend on tie-breaking.

\paragraph{Gross and net comparisons.}
The continuation values and their AI averages are
\[
\begin{array}{lrrr}
\toprule
\text{Belief or comparison} & G_\Phi & W_\Phi & K_\Phi\\
\midrule
\mu_0 & 1.112748885 & 0.628715288 & 0.484033596\\
\mu^h & 1.425000000 & 0.902500000 & 0.522500000\\
\mu^\ell & 0.300000000 & 0.090000000 & 0.210000000\\
\text{AI average} & 0.862500000 & 0.496250000 & 0.366250000\\
\text{AI minus no AI} & -0.250248885 & -0.132465288 & -0.117783596\\
\bottomrule
\end{array}
\]
Relative to the no-AI benchmark, AI advice reduces gross utility by approximately \(22.5\%\) and net utility by approximately \(21.1\%\).
Although AI lowers acquisition cost, this saving is too small to offset the loss in final-action payoff.
Gross harm can therefore occur together with net harm.